\documentclass[sn-mathphys-ay]{sn-jnl}% Math and Physical Sciences Numbered Reference Style 
\usepackage{fullpage}

\usepackage{graphicx}%
\usepackage{multirow}%
\usepackage{amsmath,amssymb,amsfonts}%
\usepackage{amsthm}%
\usepackage{mathrsfs}%
\usepackage[title]{appendix}%
\usepackage{xcolor}%
\usepackage{textcomp}%
\usepackage{manyfoot}%
\usepackage{booktabs}%
\usepackage{algorithm}%
\usepackage{algorithmicx}%
\usepackage{algpseudocode}%
\usepackage{listings}%
\usepackage{enumitem}
\newcommand{\PP}{\mathbb{P}}

\newcommand{\EE}{\mathbb{E}}
\newcommand{\E}{\text{E}}

\newtheorem{theorem}{Theorem}
\newtheorem{lemma}{Lemma}

\def \R{\mathbb{R}}
\def \E{\mathbb{E}}

\def \1{\mathds{1}}

 \def \1{{\bf 1}}

\newcommand{\m}{\mathbb}
\newcommand{\ALG}{\mathrm{ALG}}
\newcommand{\RPI}{{\normalfont\text{$k$-RPI}}}

\newtheorem{definition}{Definition}
\newtheorem{example}{Example}
\newtheorem{proposition}{Proposition}
\newtheorem{claim}{Claim}
\begin{document}

\title[Residual Prophet Inequalities]{Residual Prophet Inequalities } 


\author[1]{\fnm{Jose} \sur{Correa}}\email{correa@uchile.cl}

\author[2,3]{\fnm{Sebastian} \sur{Perez-Salazar}}\email{sperez@rice.edu}

\author[4]{\fnm{Dana} \sur{Pizarro}}\email{d.pizarro@tbs-education.fr}

\author[5]{\fnm{Bruno} \sur{Ziliotto}}\email{bruno.ziliotto@cnrs.fr}

\affil[1]{{Department of Industrial Engineering}, \orgname{Universidad de Chile}, \orgaddress{ \country{Chile}}}

\affil[2]{\orgdiv{Department of Computational Applied Mathematics and Operations Research}, \orgname{Rice University}, \orgaddress{ \country{USA}}}

\affil[3]{\orgdiv{Ken Kennedy Institute}, \orgname{Rice University}, \orgaddress{ \country{USA}}}

\affil[4]{\orgdiv{Department of Information, Operations and Decision Sciences}, \orgname{TBS Business School}, \orgaddress{ \country{France}}}

\affil[5]{\orgdiv{Toulouse School of Economics}, \orgname{Université Toulouse Capitole, Institut de Mathématiques de Toulouse, CNRS UMR 5219}, \orgaddress{ \country{France}}}
%\affil[6]{\orgdiv{Institut de Mathématiques de Toulouse}, \orgname{CNRS UMR 5219}, \orgaddress{ \country{France}}}
\abstract{We introduce a variant of the classic prophet inequality, called \emph{residual prophet inequality} ($\RPI$). In the $\RPI$
    problem, we consider a finite sequence of $n$ nonnegative independent random values with known distributions, and a known integer $0\leq k\leq n-1$. Before the gambler observes the sequence, the top $k$ values are removed, whereas the remaining $n-k$ values are streamed sequentially to the gambler. For example, one can assume that the top $k$ values have already been allocated to a higher-priority agent. Upon observing a value, the gambler must decide irrevocably whether to accept or reject it, without the possibility of revisiting past values.
    We study two variants of $\RPI$, according to whether the gambler learns online of the identity of the variable that he sees (FI model) or not (NI model).
    Our main result is a randomized algorithm in the FI model with \emph{competitive ratio} of at least $1/(k+2)$, which we show is tight.
    Our algorithm is data-driven and requires access only to the $k+1$ largest values of a single sample from the $n$ input distributions. 
    In the NI model, we provide a similar algorithm that guarantees a competitive ratio of $1/(2k+2)$.  
We further analyze independent and identically distributed instances when $k=1$. We build a single-threshold algorithm with a competitive ratio of at least 0.4901, and show that no single-threshold strategy can get a competitive ratio greater than 0.5464. }


%\pacs[History]{An extended abstract of this paper appeared in the Proceedings of the The Twenty-Sixth ACM Conference on Economics and Computation (EC'25). https://dl.acm.org/doi/10.1145/3736252.3742505}
%\pacs[]{D8, H51}


\maketitle

\section{Introduction}
    The prophet inequality is a classical model in optimal stopping theory \citep{krengel1977semiamarts,hill1982comparisons,samuel1984comparison}. In its simplest form, a finite sequence of $n$ independent and nonnegative random variables $X_1,\ldots,X_n$ is observed sequentially by a gambler. Upon observing the $i$-th value $X_i$, the gambler has to irrevocably accept the value and stop the process or reject the value and observe the next value in the sequence, if any. The gambler's goal is to devise an online algorithm that maximizes the expected accepted value. The quality of an algorithm is measured by means of the \emph{competitive ratio} which is the fraction between the expected value obtained by the algorithm and the expected optimal offline value $\E(\max_i X_i)$, the so-called prophet value. The competitive ratio thus measures the loss experienced by a gambler, who inspects the values sequentially, with respect to a prophet who knows the entire sequence of values upfront. Surprisingly, \cite{samuel1984comparison}
showed that a simple single-threshold rule guarantees a competitive ratio of at least $1/2$ and this is tight. Prophet inequalities have received renewed attention due to their applicability in posted price mechanisms and auction theory~\citep{chawla2010multi,correa2019pricing,hajiaghayi2007automated} and have become a cornerstone modeling tool for online algorithms in Bayesian scenarios and resource allocation~\citep{gallego2022constructive,goyal2023online,huang2020online}.

In this work we introduce the \emph{residual prophet inequality} (\RPI) problem: For a fixed integer $0\leq k\leq n-1$, the $k$ 
variables corresponding to the top $k$ realizations in the sequence $X_1,\ldots,X_n$ \emph{are removed before} the gambler observes the sequence. The gambler's goal is to maximize the expected accepted value among the remaining $n-k$ variables. 

The {\RPI} problem can be regarded as a \emph{robust version} of the classical prophet inequality problem (case $k=0$), where high values are impossible to obtain due to exogenous factors. The {\RPI} problem is very general and naturally relates to problems such as the postdoc problem~\citep{vanderbei2012postdoc,rose1982problem}, which have applications in hiring problems~\citep{abels2023prophet,arsenis2022individual,disser2020hiring}. Specifically, one could imagine a gambler attempting to hire an employee in a highly competitive market where the top candidates are hired by leading companies, leaving the gambler to select the best applicant from the remaining pool (see also,~\cite{perez2024robust}). Another related application is in advertising. Several platforms (e.g., YouTube, Spotify, Canva, Pandora) offer both a free version supported by ads and a premium version without ads. Essentially, users paying the premium opt out from observing ads, leaving the platform to focus on advertising the high-value users among the remaining free users.

An interesting aspect of {\RPI} concerns the information structure. Note that since some variables have been removed and the gambler will only observe the remaining ones, two different information models can be considered, depending on whether the gambler knows the identity of the observed variable at each time or not:

\medskip
{
\begin{description}
    \item[Full-information (FI)] In this version, the gambler observes the $n-k$ variables sequentially and upon observing a value, he also observes the identity (index) of the variable. \\

    \item[No-information (NI)] In this version, the gambler only observes the $n-k$ remaining values after removing the largest $k$ values.
\end{description}
}

\medskip

%The removal step poses several information models for the sequence of values observed by the gambler. We consider the following information models, ordered from the one with the most available information to the one with the least available information:

% {
% \begin{description}
%     \item[Full-information (FI)] In this version, the gambler observes the $n$ variables sequentially and upon observing a value, the value is flagged if it is one of the $k$ largest overall. Accepting a flagged value provides $0$ reward to the gambler. \\
    
%     \item[Partial-information (PI)] In this version, the gambler observes the $n$ variables sequentially, but upon observing one of the largest $k$ values---which is unknown to the gambler upfront---the gambler observes a $0$ as opposed to its actual value. \\

%     \item[No-information (NI)] In this version, the gambler only observes the $n-k$ remaining values (not even the indices) after removing the largest $k$ values.
% \end{description}
% }

% Algorithms designed for NI $\RPI$ imply algorithms for FI $\RPI$.
%, and similarly, algorithms for PI $\RPI$ imply algorithms for FI $\RPI$. 
Regardless of the information model (FI or NI), the gambler can only hope to accept a value comparable to the expectation of the largest value of the $n-k$ non-removed values. This is the expectation of the $(k+1)$-th largest value in the original sequence of $n$ values, that is, the expectation of the $(k+1)$ order statistics $\mathbb{E}(X_{(k+1)})$.\footnote{We assume that the order statistic of the variables $X_1,\ldots,X_n$ are ordered as $X_{(1)}\geq \cdots \geq X_{(n)}$. Note that this ordering is the reverse of the convention commonly used in the literature.} 

Therefore, this latter value constitutes the \emph{prophet} benchmark against which we will compare the performance of an online algorithm.
Given an information model, the competitive ratio of an algorithm for $\RPI$ is the ratio between the expected value of the algorithm and $\E(X_{(k+1)})$. Hence, a competitive ratio of $\gamma$ for NI $\RPI$ implies a competitive ratio of $\gamma$ for FI $\RPI$.
%, and similarly, a competitive ratio of $\gamma'$ for PI $\RPI$ implies a competitive ratio of $\gamma'$ for FI $\RPI$. 
Likewise, hard instances for FI $\RPI$ imply hard instances for the less informative model NI.

In contrast to the classic prophet inequality problem, the observed values in {\RPI} are correlated. The following example demonstrates that correlation plays a significant role in {\RPI}, rendering the single-threshold solutions from~\cite{samuel1984comparison} and \cite{kleinberg2012matroid} unsuitable for direct application to $\RPI$.\\

\begin{example}\label{ex:bad_thresholds}
Consider the following instance of $\RPI$ with $k=1$, $n=3$, $X_1=1$ with probability (w.p.) 1 and $X_2$ and $X_3$ both independent and identically distributed (i.i.d.) taking value $1/\varepsilon^2$ w.p.  $\varepsilon<1/2$, and $0$ otherwise.

The quantity $\EE(X_{(2)})$ is given by 
\[
\EE(X_{(2)})= \varepsilon^2 \varepsilon^{-2} + 2 \varepsilon \left(1- \varepsilon \right)= 1+ 2\varepsilon-2 \varepsilon^2.
\]

Let us analyze the performance of the strategy with the single-threshold solutions from from~\cite{samuel1984comparison} and \cite{kleinberg2012matroid},  that is  $\EE(X_{(2)})/2$ and the median of $X_{(2)}$. For the former, a case analysis shows that the gambler gets $0$ when $X_2=X_3=0$, which happens with probability  $(1-\varepsilon)^2$ and gets $1$ with the remaining probability. Thus, the gambler gets in expectation:
\[
\EE(Alg)=1-\left( 1- \varepsilon \right)^2= 2 \varepsilon-\varepsilon^2,
\]
and therefore, we have 
\[
\frac{\EE(Alg)}{\EE(X_{(2)})}= \frac{2 \varepsilon-\varepsilon^2}{1+ 2 \varepsilon-2 \varepsilon^2}<\frac{1}{2}. 
\]
Moreover, $\frac{\EE(Alg)}{\EE(X_{(2)})}= \mathcal{O}(\varepsilon)$ and then the gambler cannot guarantee a constant factor of $\EE(X_{(2)})$ using this fixed threshold. Furthermore, note that the median of $X_{(2)}$ is $0$, so using any threshold between the median and $\EE(X_{(2)})$ will not produce a different result. 
In fact, any strategy that accepts value 1 is ineffective because, once the value 1 has been observed, the expectation of the second variable is $1/((2- \varepsilon) \varepsilon)\gg 1$. Such a positive correlation between the two observed variables is what makes classic strategies fail.
\end{example}

\subsection{Results and technical contributions}

The previous example illustrates that correlation plays a major role for {\RPI}. The examples also show that traditional and well-liked thresholds such as the median or the expectation of $X_{(k+1)}$ can be arbitrarily poor choices. This is contrary to the negative correlated case where we can guarantee a competitive ratio of $1/2$~\citep{RC87,RC91,RC92}, as in the independent case. Our first contribution is a new algorithmic approach that bypasses this hardness. 

\paragraph{Main Result [Lower bound on competitive ratio]} We show that in the full information model of $\RPI$, there exists an algorithm with a competitive ratio of at least $1/(k+2)$. 
Our algorithm first samples one value from each input distribution and randomly selects one of the
$k+1$ largest values in the sample. It then uses this value and the identities of the arrivals to perform the online selection.
The randomization is independent of the input, and we note that our algorithm extends the approach of~\cite{RWW20} for the classic prophet inequality problem. We present the details of our algorithm and its analysis in Section~\ref{sec:main}.

Our algorithmic solution for the FI $\RPI$ is robust in some key aspects. On one hand, it works against any arrival order making it highly applicable in online problems. On the other, by construction, it does not need to know the distributions of each variable, but only requires one sample from each variable. This is particularly important for applications in posted price mechanisms, where consumer valuation distributions are typically unknown, and only a limited amount of past sales data is available. 
%Last, our lower bound holds irrespective of the set of $k$ variables that is removed \textcolor{red}{B: Not sure if this is still true with our new algorithm}. For example, in a company's hiring process, it is plausible that other companies may have passed on candidates that this company considers the best—perhaps because they prioritize different skills.

For the no-information model of the $\RPI$, we prove that there exists a single-threshold strategy with a competitive ratio of at least $1/(2k+2)$. The idea is similar to that of the FI $\RPI$ model, but the algorithm uses only one of the top $k+1$ sample values—selected uniformly at random—as the threshold for making the online selection.
The lack of information regarding the identities of the removed variables leads to a degradation in the competitive guarantee. Nevertheless, this guarantee can be transferred to the FI $\RPI$ model, showing that it is possible to achieve a constant-factor approximation of $\EE(X_{(k+1)})$ using a single-threshold strategy.

Next, we prove that our main result for the FI model is best possible.

\paragraph{Tightness [Upper bound on competitive ratio]} For any information model of $\RPI$, there is no algorithm that has a competitive ratio larger than $1/(k+2)$. To provide this negative result, we construct a hard instance in the FI $\RPI$ model; which will imply the negative result for the model with less information. Our instance extends the hard instance for the classic prophet inequality problem: it consists of a sequence $X_1, \ldots, X_{2(k+1)}$ of two-point distributions, where each $X_i\in \{0, a_i\}$. The values $a_i$ are positive and increase rapidly with $i$, while the events $X_i = a_i$ occur rarely for $i$ larger than $k+1$. 
We provide the details in Section~\ref{subsec:hard_instance}.

Our hard instance unveils that part of the hardness of the $\RPI$ problem stems from the \emph{order} in which the values are observed by the gambler. For every information model of $\RPI$, if the gambler observes the values in random order (RO), then there is an algorithm with a competitive ratio at least $1/e$. The algorithm is a straightforward application of the secretary problem algorithm~\cite{L61,D63,F89,gilbert1966recognizing}. 
Indeed, one can easily see that, in values are presented in random order, even in the NI $\RPI$, the standard algorithm that scans the first $(n-k)/e$ values without picking any and then takes the first value which surpasses all previously seen, guarantees a competitive ratio of at least $1/e$. To see this, let $X_1,\ldots,X_n$ be the random values and let $X_{(k+1)} \geq \cdots \geq X_{(n)}$ be the $n-k$ values observed by the gambler. From the standard analysis for the secretary problem, we are guaranteed that the gambler accepts $X_{(k+1)}$ with probability at least $1/e$. From this, the result follows. 

Our last result is an exploration of the independent and identically distributed (i.i.d.) $\RPI$ problem where $X_1,\ldots,X_n$ are drawn from the same distribution.

\paragraph{Additional result [i.i.d. $\RPI$, $k=1$]} 
The previous observation implies that for i.i.d. random variables and arbitrary $k$, we can always guarantee a factor $1/e$ using the classic secretary algorithm.  This shows a stark difference between $\RPI$ and its i.i.d. counterpart and indeed even for small $k$ this improves upon our tight factor of $1/(k+2)$ for $\RPI$. Therefore, it is interesting to explore the gap between the i.i.d. and the independent versions of the problem, even in the case $k=1$.
We prove that for both information models of 1-RPI, if the values \(X_1, \ldots, X_n\) are i.i.d., then there exists an algorithm with a competitive ratio \(0.4901\). Our algorithm here is more standard. We propose a single-threshold strategy for NI 1-RPI, which determines the threshold \(\tau\) via \(\Pr(X \geq \tau) = q\), where \(q\) is an input quantile. Our analysis follows a quantile-based approach, expressing both the expected value of the algorithm and the optimal value \(\E (X_{(2)})\) as functions of quantiles. By comparing their ratio, we derive a lower bound that depends solely on \(q\). Optimizing over \(q\) yields the desired result. We also establish that no single-threshold strategy can achieve a competitive ratio greater than \(0.5464\) in any information model. This result shows that the optimal competitive ratio of \(1 - 1/e\) for single-threshold strategies~\citep{correa2021posted,hill1982comparisons}, attained when \(k = 0\), cannot be recovered for \(k \geq 1\). We present the details in Section~\ref{sect:iid}.

\subsection{Related Literature}

The prophet inequality problem, as introduced by \citet{krengel1977semiamarts}, was resolved by using a dynamic program that gave a tight approximation ratio of $1/2$. 
\cite{samuel1984comparison} later proved that a single-threshold strategy yields the same guarantee; this also showed that the order in which the variables are observed is immaterial. The renewed interest in prophet inequalities is due to their relevance to auctions, specifically posted priced mechanisms (PPMs) in online sales~\citep{Alaei2011,Chawla2010,Duetting2020,Hajiaghayi2007,Kleinberg2012}. It was implicitly shown by \cite{Chawla2010} and \cite{Hajiaghayi2007} that every prophet-type inequality implies a corresponding approximation guarantee in a PPM, and the converse is true as well~\citep{correa2019pricing}.

The closest work to ours is likely that of~\cite{RWW20}, where the authors used the principle of deferred decision to prove that a single sample from each distribution is sufficient to achieve a competitive ratio of \(1/2\) for the classic prophet inequality. This technique has also been applied to other optimal stopping problems (see, e.g.,~\cite{correa2022two,nuti2023secretary}).  

In essence, after obtaining one sample from each distribution,~\cite{RWW20} sets the threshold as the maximum of these samples. Although our proof for the general case is also based on this principle, the analysis is much more intricate due to the complexity of the $\RPI$ problem, which necessitates a more sophisticated algorithm. Specifically, for our approach to be effective, it is insufficient to simply use a threshold based on the $j$-th order statistic of the sample set for some fixed $j$. 
Instead, the algorithm first selects $j$ according to a carefully chosen distribution. Moreover, in the FI model, the algorithm must discard certain elements based on their identity, even when their values exceed the $j$-th order statistic.
Thus, in contrast to Rubinstein’s work, where $j$ is deterministically fixed at 1, our approach introduces an additional layer of randomization, and the $j$-th order statistic is not exactly used as a threshold in the FI model.

There has been a growing interest in competitive versions of online selection problems \citep{ezra2021prophet, GPR24, IKM06,KL15, R24}. The closest paper in this stream of literature to ours is the one by \cite{ezra2021prophet}, where the authors consider a generalization of the prophet inequality problem with $k+1$ gamblers. Gambler $j$ observes the sequence after the first $j-1$ gamblers have gone through the sequence, and they study reward guarantees
under single-threshold strategies.  Note that, in our case, we can imagine that there are $k+1$ gamblers but the first $k$ gamblers are all-mighty. These $k$ gamblers are not strategic, hence we do not need a game-theoretic analysis, unlike in the aforementioned papers on competitive prophet inequalities.

\section{Model}

For $0\leq k \leq n-1$, an instance of {\RPI} is given by a sequence $X_1, \dots, X_n$ of nonnegative independent random variables, where $X_i$ has cumulative density function (cdf) $F_i$. \emph{Nature} removes $k$ variables corresponding to the top $k$ realizations,\footnote{If there are several choices due to ties, Nature randomizes the choice of the $k$ variables.} and we denote by $D$ the corresponding set of indices of the remaining variables.
We consider two information models that determines what the gambler observes sequentially. 
%There are three possible information models that determine how the gambler observes the sequence of values.

In the \emph{full information} (FI) model, the gambler observes online the pairs $ (X_i, i)_{i \in D}$. That is, the gambler observes both the value and the index of the random variable from which the value originates.
%$( (X_i, \mathbf{1}_D (i) )_{i=1,\ldots,n}$ where $\mathbf{1}_D(\cdot)$ is the indicator function of the set $D$ which is $1$ if and only if $i\in D$. That is, in the FI model, the indicator flags the $k$ largest values removed by nature. 

%In the \emph{partial information} (PI) model, the gambler observes online the values in the sequence $(X_i \cdot \mathbf{1}_{[n]\setminus D}(i))_{i=1,\ldots,n}$, where we replace $X_i$ with $i\in D$ to be $0$. 
In the \emph{no information} (NI) model, the gambler only observes online the $n-k$ values in the sequence $(X_i)_{i\in D}$. In both information models, $D$ is unknown to the gambler upfront. Given an information model (FI or NI), the gambler wants to implement an online algorithm $\ALG$ that observes the online values according to the information model and accepts a value. %For the FI and PI models, the value accepted is $0$ if it is a value with index in $D$. 
Regardless of the model, and abusing notation, we denote by $\ALG$ the value accepted by the online algorithm. The expected optimal offline solution corresponds to $\EE\left( \max_{i\in D} X_i \right)=\EE(  X_{(k+1)})$. 

For $\gamma>0$, we say that $\ALG$ has a \emph{competitive ratio} $\gamma$ if $\EE(\ALG)\geq \gamma \cdot \EE(X_{(k+1)})$ for any input of $\RPI$. For each $k$, we are interested in finding the largest $\gamma_k$ such that there is an algorithm $\ALG$ with competitive ratio $\gamma_k$ for $\RPI$. Note that for $k=0$, we have $\gamma_0=1/2$~\citep{samuel1984comparison}. We note that an algorithm with a competitive ratio $\gamma$ for the NI model implies an algorithm with competitive ratio $\gamma$ for the FI model.

%We note that an algorithm with a competitive ratio $\gamma$ for the NI model implies an algorithm with competitive ratio $\gamma$ for the PI model, and similarly, an algorithm with competitive ratio $\gamma'$ for the PI model implies an algorithm with competitive ratio $\gamma'$ for the FI model.

\section{Lower bound on competitive ratio} \label{sec:main}

In this section, we prove our main result.
%and its tightness. 
We assume that the distributions $F_1,\ldots ,F_n$ are independent but not necessarily identically distributed. \\

\begin{theorem} 
 \label{thm:gen1}
    For the FI model, there is an algorithm for $\RPI$ with competitive ratio at least $1/(k+2)$. \\

    \end{theorem} 

\begin{theorem} 
 \label{thm:gen1'}
    For the NI model, there is a single-threshold algorithm for $\RPI$ with competitive ratio at least $1/(2k+2)$. \\

    \end{theorem} 

%\textcolor{red}{B: Our algorithm breaks ties at random, so is it really a ``single-threshold'' algorithm?\\
%Jose: I think yes. generically there are no ties but if they are normally single threshold allows for random tie braking. For instance in iid PI we say that the best possible with a single threshold is 1-1/e, but actually this is not true if random tie breaking I snot allowed.}

%To prove Theorem~\ref{thm:gen1}, we provide an algorithm with competitive ratio $(k+2)^{-1}$ for PI $\RPI$. This will imply the result for the model with more information FI. 
To prove both Theorem~\ref{thm:gen1} and Theorem~\ref{thm:gen1'}, we employ a randomized strategy. In the case of Theorem~\ref{thm:gen1'}, the strategy is, in fact, a randomized threshold strategy. We highlight here that, as a corollary of Theorem~\ref{thm:gen1'}, we obtain that in the FI model, there exists a threshold strategy with a competitive ratio of at least 
$\frac{1}{2(k+1)}$.

%\subsection{A randomized threshold strategy} \label{sec:thm1:1}

To understand the rationale behind the construction of our randomized strategies to prove Theorem~\ref{thm:gen1} and Theorem~\ref{thm:gen1'}, let us recall the result obtained by~\cite{RWW20} in the classic prophet inequality setting. By drawing one sample from each distribution and taking the maximum of them as a threshold, the gambler can guarantee a competitive ratio of $1/2$. A natural adaptation of that algorithm to our setting is to consider as a threshold the \textit{$(k+1)$-th} maximum of the samples. We denote by $MSA_{k+1}$ such a strategy. 
 
Unfortunately, such a strategy does not guarantee any constant competitive ratio. Indeed, consider again the instance in Example~\ref{ex:bad_thresholds} with $k=1$.

The expected value of the algorithm $MSA_2$ is:
\begin{eqnarray*}
   \EE(MSA_2)&=&\varepsilon^{2} \EE(MSA_2 | \tau=\varepsilon^{-2})+ (1-\varepsilon^{2}) \EE(MSA_2|\tau \leq 1) \\
   &=& \varepsilon^{2} \cdot \varepsilon^{2} \cdot \varepsilon^{-2}+(1-\varepsilon^{2}) \cdot
   [1-(1-\varepsilon)^2]\cdot 1
   \\
   &=& 2 \varepsilon-2 \varepsilon^{3}+\varepsilon^{4}.
\end{eqnarray*}
Given that $\EE(X_{(2)})\rightarrow 1$ when $\varepsilon \rightarrow 0$, we obtain that $\EE(MSA_2)/\EE(X_{(2)}) \rightarrow 0$ as $\varepsilon$ tends to zero.

To tackle this problem and establish the competitive ratio stated in Theorem~\ref{thm:gen1'}, we draw one sample $s_i$ from each distribution $F_i$ and consider the following $k+1$ algorithms. \\

\begin{definition}
Given $i \in \{1, \dots, k+1\}$, $MSA_i$ is the strategy proceeding as follows:

\begin{enumerate}
    \item Draw one independent sample $s_j \sim F_j$ for each $j = 1, \dots, n$.
%    \item Sort the sample values in decreasing order.
    \item Let $\tau$ be the $i$-th largest value among the samples.
    %\item Observe the sequence of remaining values, without knowing their identities.
    \item Select the first value $x_t$ such that $x_t$ is higher than $\tau$. \\
\end{enumerate}
\end{definition}

On the other hand, to prove Theorem~\ref{thm:gen1}, we make use of the algorithm $MSA_{k+1}$ defined above, along with the following $k$ algorithms.\\

\begin{definition}
The strategy $\overline{\text{MSA}}_i$, for $i \in \{1, \dots, k\}$, proceeds as follows:
\begin{enumerate}
    \item Draw one independent sample $s_j \sim F_j$ for each $j = 1, \dots, n$.
   % \item Sort the sample values in decreasing order.
    \item Let $\tau$ be the $i$-th largest sample value, and let $j^*$ be the index of the distribution from which that sample came.
   % \item Observe the remaining values one by one, along with their identities.
    \item Select the first value $x_t$ such that:
    \begin{itemize}
        \item $x_t$ is higher than $\tau$, and
        \item $x_t$ does not come from distribution $F_{j^*}$.
    \end{itemize}
\end{enumerate}
\end{definition}

For both algorithms and in the case where there are equalities between samples or between the threshold and the observed value, we break ties at random.  Note that the algorithms $\overline{MSA}_{i}$, for $i \in \{1, \dots, k\}$, must determine whether the arriving value originates from the same distribution as the sample used to define the threshold, and therefore, the knowledge of the identity of each variable is necessary for the online selection. A complete analysis of these algorithms is provided in Sections~\ref{sec:thm1:2} and~\ref{sec:thm1:1'}. 

By the principle of deferred decision and following the formalism in~\cite{RWW20}, instead of considering one sample for each distribution and then looking at the real values in an online fashion, we can draw two samples from each distribution $F_i$, namely $y_i$ and $z_i$, and then flip a fair coin to decide which is equal to $s_i$ and which is equal to $x_i$. This procedure correctly generates $s_1, \dots, s_n$ and $x_1, \dots, x_n$ as independent draws of $F_1, \dots F_n$. From now on, we will denote by $S$ the set of samples $\{s_1, \dots, s_n \}$ and $X$ the set of true values $\{ x_1, \dots, x_n\}$.

To analyze the performance of the algorithms, we assume that for each $i$, $y_i>z_i$ and we order all these samples in decreasing order, relabeling them as $w_1, \dots, w_{2n}$, so that $w_1 \geq w_2 \geq w_3 \geq \dots \geq w_{2n}$\footnote{If some values are identical, Nature randomizes their order within the sequence.}. We say that $(w_l,w_{l'})$ is a pair, or that $w_l$ is paired with $w_{l'}$, if they originate from the same distribution.

Moreover, for each $j \in  \{1, \dots, k+1\}$ we define $\xi_j$ as the corresponding position of the $j$-th value $z$ in the sequence of $w's$ values. For example, if the first elements of the $w$ sequence are given by 
\[
y_3 \ \ y_5 \ \ y_1 \ \ z_5 \ \ y_8 \ \ z_8 \ \ z_3 \ \ \dots, 
\]
then $\xi_1=4$ and $\xi_2=6$. Note that $\xi_j$ can also be seen as the position at which the $j$-th pair $(y,z)$ from the same distribution appears. In the subsequent analysis, we fix specific realizations of the pairs $(y_i,z_i)$, which in turn determine the $\xi_j$ and the $w_i$. 

\subsection{Proof of Theorem \ref{thm:gen1}} \label{sec:thm1:2}

To show Theorem \ref{thm:gen1}, we consider the $k+1$ algorithms $\overline{MSA}_1, \dots, \overline{MSA}_k$, $MSA_{k+1}$ defined in Section~\ref{sec:main}, and use them to define the randomized strategy $\overline{MSA}_{RAND}$ as follows:
\begin{itemize}
    \item[(1)]Before the game starts, select a random number $I$ in 
 $\left\{1,\dots,k+1\right\}$, such that for all $i \in \left\{1,\dots,k\right\}$, $I=i$ with probability $1/(k+2)$, and $I=k+1$ with probability $2/(k+2)$. 
 \item[(2)] Play $\overline{MSA}_I$, if $I \in \{1, \dots, k\}$, and $MSA_{k+1}$, if $I=k+1$.
\end{itemize}

 We prove Theorem~\ref{thm:gen1} by showing that the strategy $\overline{MSA}_{RAND}$ has a competitive ratio $\frac{1}{k+2}$. Before proceeding to the proof of Theorem~\ref{thm:gen1}, we need to introduce some definitions and two technical lemmas, which we prove later.
 
\begin{definition}
Let $ l \in \{1, \dots, 2k+1\} $. We say that $w_l$ is \emph{blocked} if there exist $ r, r' \in \{l+1, \dots, 2k+1\}$ such that $w_{r'} = y_j$ and $w_r = z_j$ for some $ j $. We denote by $m_l$ the smallest $r$ that satisfies this property. 
\end{definition}

For example, if $k=3$ and the first $2k+1=7$ elements of the $w$ sequence are given by
\[
y_3 \ \ y_5 \ \ y_1 \ \ z_5 \ \ y_8 \ \ z_8 \ \ z_1 \ \ \dots, 
\]
$w_2$ is blocked, since the pairs $(y_1,z_1)$ and $(y_8,z_8)$ appear between the $3$-rd and $7$-th positions. Moreover, in this case $m_2=6$.

The pair $(y_j,z_j)$ ``blocks'' $w_l$, in the sense that no matter whether $z_j=w_r$ is in $X$ or $S$, no threshold below $w_r$ can guarantee selecting the value $w_l$. \\

\begin{definition}
Let $l \in \left\{1,\dots,2k+1\right\}$ and $p$ such that $w_p$ is paired with $w_l$. We say that $w_l$ is \textit{ill-paired} if $p \in \left\{l+1,\dots,2k+1\right\}$. \\ %such that $w_l$ is paired with an element $w_r$. We call $p$ the smallest $r$ satisfying such a property. 
\end{definition} 

That is, we say that a value $w_l$ is ill-paired if it is paired with a value greater than or equal to $w_{2k+1}$. For instance, considering the same sequence as before, $w_2$ is ill-paired since $z_5$ appears before $w_7$. \\

%The lemmas establish a lower bound on the expected value attained by the gambler when following the algorithms  $MSA_i$, expressed as a function of the probability density function of  $X_{(k+1)}$. \\
\begin{definition}
    For each $l \in \{1, \dots, 2k+1\}$, we define the parameter $\delta_l$ as follows:

\[ \delta_l=  \left\{
\begin{array}{ll}
      2^{-2k+l-1} & \text{ if } w_l \text{ is not blocked and not ill-paired} \\
     2^{-2k+l} & \text{ if } w_l \text{ is not blocked and ill-paired} \\
      0 & \text{ otherwise.}\\
\end{array} 
\right. \] \\

\end{definition}
%Before proceeding to the proof of Theorem~\ref{thm:gen1}, we need to introduce two technical lemmas, whose proofs have been deferred to Appendix~\ref{sec:thm1:3}. \\

\begin{proposition} \label{lem:bound:msak+1}
        If the gambler plays according to $MSA_{k+1}$, his expected reward is at least

 \begin{eqnarray}\label{ineq:MSA_k}  \nonumber
      \EE(MSA_{k+1}) &\ge& 
 \frac{1}{2} \sum_{l=k+1}^{2k+1} \PP(X_{(k+1)}=w_l) w_l \delta_l
 + \frac{1}{2} \sum_{l=2(k+1)}^{\xi_{k+1}} \PP(X_{(k+1)}=w_l) w_l
  \end{eqnarray}

\end{proposition}

% \begin{lemma} \label{lem:msak}
%         If the gambler plays according to $MSA_{k+1}$, his expected reward conditional on $\xi_j$ for  $j \in \{1, \dots, k+1\}$ is at least 

%  \begin{eqnarray}\label{ineq:MSA_k}  \nonumber
%       \EE(MSA_{k+1}) &\ge& 
%  \sum_{l=k+1}^{2k} \PP(X_{(k+1)}=w_l) 2^{l-2(k+1)
%  } w_{l} + \frac{1}{2} \sum_{l=2k+1}^{\xi_{k+1}} \PP(X_{(k+1)}=w_l) w_l
%   \end{eqnarray}
% \end{lemma}
% \textcolor{blue}{This should be $xi_{k+1}-1$ instead of $\xi_{k+1}$ in the above lemma}

Note that $\EE(X_{(k+1)})$ is equal to $\sum_{l=k+1}^{\xi_{k+1}} w_l \PP(X_{(k+1)}=w_l)$. Consequently, when $k=0$, Proposition~\ref{lem:bound:msak+1} recovers the result from~\cite{RWW20} which states that $MSA_1$ gives a $1/2$ competitive ratio ($\delta_1=1/2$ since $w_1$ is neither blocked nor ill-paired).  The challenge when $k \geq 1$ arises from the fact that, for $k+1\leq l \leq 2k+1$, the coefficient accompanying the term $\PP(X_{(k+1)}=w_l) w_l$ may be smaller than $1/(k+2)$. In other words, for $k+1 \leq l \leq 2k+1$ the coefficient may be ``too small'', while for $2k+2 \leq l \leq \xi_{k+1}$, it is ``larger than necessary'' (equal to $1/2$). This imbalance motivates the introduction of a randomization over the $MSA_{k+1}$ and $\overline{MSA}_i$ algorithms: By blending $MSA_{k+1}$ with $\overline{MSA}_i$ for $i \in \left\{1,\dots,k\right\}$, we redistribute these coefficients more evenly. 
To analyze such a randomization, we need first a lower bound on the performance of $\overline{MSA}_i, i \in \left\{1,\dots,k\right\}$. \\

\begin{proposition}
    \label{lem:bound:msai}
The sum of the expected reward of the gambler playing according to $\overline{MSA}_i$ for $i \leq k$ is at least 
    \begin{eqnarray*}
\sum_{i=1}^{k} \m{E}(\overline{MSA}_i) \geq  \sum_{l=1}^{2k+1} \m{P}(X_{(k+1)}=w_l) w_l (1- \delta_l).
\end{eqnarray*}

\end{proposition}

% \begin{lemma} \label{lem:msai}
%     If the gambler plays according to $MSA_i$ for any $i \in \{1, \dots, k\}$, his expected reward conditional on $\xi_j$ for  $j \in \{1, \dots, k+1\}$ is at least 

%  \begin{eqnarray}\label{ineq:MSA_i}  \nonumber
%       \EE(MSA_i) &\ge& 
%  \sum_{l=k+1}^{k+i} \PP(X_{(k+1)}=w_l) 2^{l-k-i-1} w_{l} 
%   \end{eqnarray}
% \end{lemma}

The coefficients accompanying the $\PP(X_{(k+1)}=w_l) w_l$ in the above inequality are higher than those in the expression of Proposition \ref{lem:bound:msak+1} for $k+1 \leq l \leq 2k$, while they are equal to 0 for $l>2k+1$. This supports the idea that combining algorithms enables a redistribution of coefficients. The surprising fact is that there exists a way to combine the $\overline{MSA}_i, i \in \left\{1,\dots,k+1\right\}$ in a way that all the coefficients are simultaneously higher than $1/(k+2)$, yielding the competitive factor of $1/(k+2)$. We prove this below.

\begin{proof}[Proof of Theorem \ref{thm:gen1}]
Let us consider the strategy for the gambler $\overline{MSA}_{RAND}$ consisting on playing according to $\overline{MSA}_i$ with probability $1/(k+2)$, for $i \in \{1, \dots, k\}$, and to $MSA_{k+1}$ with probability $2/(k+2)$. 

Then, $(k+2) \EE(\overline{MSA}_{RAND})= \sum_{i=1}^{k} \EE(\overline{MSA}_i)+ 2 \EE(MSA_{k+1}), $ and 
by using Proposition~\ref{lem:bound:msak+1} and Proposition~\ref{lem:bound:msai}, we obtain
% By definition of $\sigma$, 
% $\alpha \EE(\sigma)= \sum_{i=1}^{k+1} \alpha_i \EE(MSA_i), $ and 
% by using Lemma~\ref{lem:msak} and Lemma~\ref{lem:msai}, we obtain
\begin{eqnarray*}\label{ineq:prop_1} \nonumber
  (k+2) \EE(\overline{MSA}_{RAND})&\geq& \sum_{l=1}^{2k+1}  \PP(X_{(k+1)}=w_l) w_l (1-\delta_l) \\
   &+&  2 \sum_{l=k+1}^{2k+1}  \PP(X_{(k+1)}=w_l) w_l \frac{\delta_l}{2} + 2 \sum_{l=2(k+1)}^{\xi_{k+1}}  \PP(X_{(k+1)}=w_l) w_l \frac{1}{2} \\ \nonumber
   &=& \sum_{l=k+1}^{\xi_{k+1}}  \PP(X_{(k+1)}=w_l) w_l = \EE(X_{(k+1)}),
\end{eqnarray*}
where the equality holds because $\PP(X_{(k+1)}=w_l)=0$ for $l<k+1$. 
This concludes on the proof of Theorem~\ref{thm:gen1}.
\end{proof}

%\section{Missing Proofs From Section~\ref{sec:main}}
%\label{sec:thm1:3}

\subsubsection{Proof of Proposition \ref{lem:bound:msak+1}}
% Recall that Proposition \ref{lem:bound:msak+1} states the following inequality:
% \begin{eqnarray}  \nonumber
%       \EE(MSA_{k+1}) &\ge& 
%  \frac{1}{2} \sum_{l=k+1}^{2k+1} \PP(X_{(k+1)}=w_l) w_l \delta_l
%  + \frac{1}{2} \sum_{l=2(k+1)}^{\xi_{k+1}} \PP(X_{(k+1)}=w_l) w_l.
%   \end{eqnarray}
  
The proof of Proposition \ref{lem:bound:msak+1} is divided into two intermediary results, which we state now. \\

\begin{lemma} \label{lem:bound:msak+1:lem1}
\begin{equation*}
\EE(MSA_{k+1} 1_{\tau_{k+1}=w_{2k+2}}) \ge 
 \sum_{l=k+1}^{2k+1} \PP(X_{(k+1)}=w_l) w_l \frac{\delta_l}{2}. \\
 \end{equation*}
\end{lemma}
\begin{lemma} \label{lem:bound:msak+1:lem2}
Assume that $2k+2 \neq \xi_{k+1}$. Then
\begin{equation*}
\EE(MSA_{k+1} 1_{\tau_{k+1} \leq w_{2k+3}}) \ge 
 \frac{1}{2} \sum_{l=2(k+1)}^{\xi_{k+1}} \PP(X_{(k+1)}=w_l) w_l.
 \end{equation*}
\end{lemma}
\begin{proof}[Proof of Proposition \ref{lem:bound:msak+1} admitting Lemmas \ref{lem:bound:msak+1:lem1} and \ref{lem:bound:msak+1:lem2}]
In the case where $2k+2 \neq \xi_{k+1}$,
summing the two inequalities proves Proposition \ref{lem:bound:msak+1}. Assume that $2k+2=\xi_{k+1}$. This means that the elements of $\left\{w_1,\dots,w_{2k+2}\right\}$ form $k+1$ pairs. Hence, if $w_{2k+2} \in S$, which happens with probability $1/2$, there are exactly $k+1$ elements larger than $w_{2k+2}$ that are in $X$. In that case, $MSA_{k+1}$ picks $X_{(k+1)}$. It follows that 
\begin{equation*}
\m{E}(MSA_{k+1}) \geq \frac{1}{2} \m{E}(X_{(k+1)}).
\end{equation*}
In particular, Proposition \ref{lem:bound:msak+1} holds. 
\end{proof}

\subsubsection*{Proof of Lemma \ref{lem:bound:msak+1:lem1}}

Lemma \ref{lem:bound:msak+1:lem1} is a consequence of the following lemma. \\

\begin{lemma} \label{ineq:MSA:2}
Let $l \in \{1, \dots, 2k+1\}$ such that $w_l$ is not blocked.

\begin{enumerate}[label=\alph*)]
    \item 
    \label{ineq:MSA:2a}
    If $w_l$ is not ill-paired, it holds that
    \[
    \PP( \{MSA_{k+1}=w_l \} \cap \{ \tau_{k+1}=w_{2k+2} \}) | X_{(k+1)}=w_l) \geq 2^{-2k-2+l} .
    \]
  \item 
  \label{ineq:MSA:2b}
  If $w_l$ is ill-paired, then
     \[    
   \PP(\{ MSA_{k+1}=w_l \} \cap \{ \tau_{k+1}=w_{2k+2} \} | X_{(k+1)}=w_l) \geq  
2^{-2k-1+l}.
\]
\end{enumerate}
\end{lemma}
\begin{proof}[Proof of Lemma~\ref{ineq:MSA:2}]
    \begin{enumerate}
    \item[$a)$]
 %First, note that since $l$ is not ill-paired, $w_l$ is not paired with $w_{2k+2}$. Then, 
  %Take $i \in \{l-k, \dots, k\}$. 
  %We want to analyze
% $\m{P}(MSA_i=w_l| X_{(k+1)}=w_l) $.
We claim that when $X_{(k+1)}=w_l$ and $\tau_{k+1}=w_{2k+2}$, then $MSA_{k+1}$ picks $w_l$. Indeeed, when $X_{(k+1)}=w_l$, there are exactly $l-1-k$ elements in $\left\{w_1,\dots,w_{l-1}\right\}$ that are in $S$. If, in addition, $\tau_{k+1}=w_{2k+2}$, then there should be exactly $k-(l-1-k)=2k+1-l$ elements of $\left\{w_{l+1},\dots,w_{2k+1}\right\}$ that are in $S$, meaning that they should all be in $S$. Under these circumstances, $w_l$ is the only element in $X$ that is above $\tau_{2k+2}$ and that is not among the $k$ best values in $X$, and is thus selected by $MSA_{k+1}$. 
%Last, $l$ is not ill-paired, hence $w_l$ is not paired with $w_{2k+2}$. 
We deduce that  
\begin{eqnarray*}
      \m{P}(\left\{MSA_{k+1}=w_l\right\} \cap \left\{\tau_{k+1}=w_{2k+2}\right\} |X_{(k+1)}=w_l) \\ 
      =
       \m{P}(\tau_{k+1}=w_{2k+2}|X_{(k+1)}=w_l).
%      \\
%      &=&
%      \m{P}(\tau_i=w_{2k+2}|X_{(k+1)}=w_l) \m{P}(X_{(k+1)}=w_l) w_l
 %\sum_{l=k+1, l \in G^{i}}^{k+i} \PP(X_{(k+1)}=w_l) 2^{l-k-i-1} w_{l} 
  \end{eqnarray*}
  
  Therefore, it is enough to prove $ \m{P}(\tau_{k+1}=w_{2k+2}|X_{(k+1)}=w_l) \geq 2^{-2k-2+l} . $
  %there are $k$ elements below $w_l$ that are in $X$, and $i-1$ elements below $w_{2k+2}$ that are in $S$, hence 
  Given $X_{(k+1)}=w_l$, in order for $\tau_{k+1}=w_{2k+2}$ to hold, it is necessary and sufficient that all the elements in $\left\{w_{l+1},\dots,w_{2k+2}\right\}$ belong to $S$. We claim that this event occurs with probability greater than $2^{-2k+l-2}$. 
  To show that, we use the chain rule for conditional probability:
  \begin{multline*}
       \PP(\left\{w_{l+1},\dots,w_{2k+2}\right\} \subset S | X_{k+1}=w_l ) =\PP \left( \left. \bigcap_{j=l+1}^{2k+2} \{w_j \in S\} \right  \vert X_{(k+1)}=w_l \right) \\
  =\prod_{l'=l+1}^{2k+2} \PP \left( \left. \{w_{l'} \in S\}  \right  \vert \bigcap_{j=l+1}^{l'-1} \{w_j \in S\},  X_{(k+1)}=w_l \right).
\end{multline*}
In order to establish the desired result, it is sufficient to verify that each factor in the expression above is lower bounded by $1/2$. The proof is therefore divided into two steps:\\

 \noindent{\bf Step 1:} $ \m{P}(w_{l+1} \in S|X_{(k+1)}=w_l) \geq 1/2.$ \\

If $w_{l+1}$ is paired with an element smaller than $w_{l+1}$, then the events $\left\{w_{l+1} \in S \right\}$ and $\left\{X_{(k+1)}=w_l \right\}$ are independent, and therefore
\[ \m{P}(w_{l+1} \in S|X_{(k+1)}=w_l) = 1/2. \]

  Consider now the case where $w_{l+1}$ is paired with some $w_a \geq w_{l+1}$. Since $w_l$ is not ill-paired, we have $a \neq l$, and the probability that $w_{l+1}$ lies in $S$ is equal to the probability that $w_a$ is one of the not-paired elements of $\{w_1, \dots, w_l \}$ in $X$. Note that the event $\{X_{(k+1)}=w_l \}$ occurs if and only if $w_l \in X$ and there are exactly $k$ elements in X that are larger than $w_l$. Therefore, if $l \in \{ \xi_j, \dots, \xi_{j+1}-1\}$, then among the $l-2j$ not-paired values in $\{w_1, \dots, w_l \}$,  $k+1-j$ belong to $X$, while $l-k-1-j$ are in S. It follows that the probability of $w_a$ being among those elements in $X$ is higher than $1/2$, since $k+1-j> l-k-1-j$ due to $l \leq 2k+1$. We thus conclude that 
   \begin{eqnarray*}
  \m{P}(w_{l+1} \in S|X_{k+1}=w_l) > \frac{1}{2}. \\
  \end{eqnarray*}
  
  % Call $b$ the number of pairs that are included in $\left\{1,\dots,l \right\}$. 
  % Call $C \subset \left\{1,\dots,l\right\}$ the set of numbers that are not in such pairs. Then, 
  % \begin{eqnarray*}
  % \m{P}(w_{l+1} \in S|X_{k+1}=w_l)&=& \m{P}(w_{a} \in X | \text{k+1-b elements of C are in X and l-1-k-b elements are in S})
  % \\
  % &\geq& \frac{1}{2}. 
  % \end{eqnarray*}
  
 \noindent{\bf Step 2:} For each $l' \in \left\{l+2,\dots,2k+2\right\}$,  
\begin{equation*}
\m{P} \left( w_{l'} \in S|\bigcap_{j=l+1}^{l'-1} \{w_j \in S\}, X_{(k+1)}=w_l \right) \geq \frac{1}{2}.
\end{equation*}

Let $w_a$ such that $w_{l'}$ is paired with $w_a$, and assume that ${l'} \in \{ \xi_j+1, \dots, \xi_{j+1} \}$. That is, there are $j$ pairs that arrived before $w_{l'}
$. 
Following the same argument than in Step 1, if $w_a<w_{l'}$, we have 
\begin{equation*}
\m{P} \left( w_{l'} \in S| \bigcap_{j=l+1}^{l'-1} \{w_j \in S\}, X_{(k+1)}=w_l \right) = \frac{1}{2}.
\end{equation*}

On the other hand, $w_a$ cannot belong to $\{w_{l+1}, \dots, w_{l'-1}\}$ because $w_l$ is not blocked. 

Finally, let us assume $w_a \leq w_l$. In this case, among the $l'-1-2j$ not-paired values in $\{w_1, \dots, w_{l'-1} \}$,  $k+1-j$ belong to $X$, while $l'-2-j-k$ are in S. Then, it follows that the probability of $w_a$ being among those elements in $X$ is higher than $1/2$, since $k+1-j> l'-2-j-k$ due to $l' \leq 2k+1$. We thus conclude that

\begin{equation*}
\m{P} \left( w_{l'} \in S|  \bigcap_{j=l+1}^{l'-1} \{w_j \in S\},X_{(k+1)}=w_l \right) > \frac{1}{2}.
\end{equation*}

% Call $b$ the number of pairs that are included in $\left\{1,\dots,l'-1\right\}$, and $C \subset \left\{1,\dots,l'-1\right\}$ the set of numbers that are not in such pairs. Then
% \begin{eqnarray*}
% \m{P}(w_{l'} \in S| X_{(k+1)}=w_l,\left\{w_{l+1},\dots,w_{l'-1} \right\} \subset S) &=&
% \m{P}(w_a \in X| \text{k+1-b elements of C are in X and l'-1-k-b elements are in S})
% \\
% &\geq& \frac{1}{2},
% \end{eqnarray*}
% where the last inequality stems from the fact that $k+1-b \geq l'-1-k-b$. 
% %  $X_{k+1}=w_l$ if and only if, there are exactly $k-b$ elements of $C$ that are in $X$ and 
Combining Step 1 and Step 2 yields the result. \\

\item[$b)$]
%Take $i \in \left\{p-k,\dots,k\right\}$. In this case, $k+i+1>p$, and then $w_l$ is not paired with $w_{k+i+1}$. 
As in the proof of Case \ref{ineq:MSA:2a}, we have 
     \begin{eqnarray*}
      \m{P}(\{MSA_{k+1}=w_l \} \cap \left\{\tau_{k+1}=w_{2k+2}\right\} |X_{(k+1)}=w_l) 
      = 
       \m{P}(\tau_{k+1}=w_{2k+2}|X_{(k+1)}=w_l),
%      \\
%      &=&
%      \m{P}(\tau_i=w_{k+i+1}|X_{(k+1)}=w_l) \m{P}(X_{(k+1)}=w_l) w_l
 %\sum_{l=k+1, l \in G^{i}}^{k+i} \PP(X_{(k+1)}=w_l) 2^{l-k-i-1} w_{l} 
  \end{eqnarray*}

  and to obtain the result it is enough to show that 
  \begin{equation} \label{ineq:MSAk:inter}
 \m{P}(\tau_{k+1}=w_{2k+2}|X_{(k+1)}=w_l) \geq 2^{-2k-1+l}.
  \end{equation}
  
   Given that $X_{(k+1)}=w_l$, we know that $w_p$ is in $S$, since $(w_p,w_l)$ is a pair. Then, in order to get $\tau_{k+1}=w_{2k+2}$, it is necessary and sufficient that all the elements in $\left\{w_{l+1},\dots,w_{2k+2}\right\} \setminus \left\{p\right\}$ belong to $S$. This happens with probability at least $2^{-2k-1+l}$, by the same argument as in the proof of Case \ref{ineq:MSA:2a}. This proves \eqref{ineq:MSAk:inter}, and the result follows. 
    \end{enumerate}
\end{proof}

\begin{proof}[Proof of Lemma \ref{lem:bound:msak+1:lem1}]

    We want to prove
\begin{equation*}
\EE(MSA_{k+1} 1_{\tau_{k+1}=w_{2k+2}}) \ge 
 \sum_{l=k+1}^{2k+1} \PP(X_{(k+1)}=w_l) w_l \frac{\delta_l}{2}. 
 \end{equation*}

 To this end, note that 
\[
\EE(MSA_{k+1} 1_{\tau_{k+1}=w_{2k+2}})= \sum_{l=k+1}^{2k+1} w_l       \m{P}(\{MSA_{k+1}=w_l \} \cap \left\{\tau_{k+1}=w_{2k+2}\right\} |X_{(k+1)}=w_l) \PP(X_{(k+1)}=w_l).
\]

% We need to study $\m{P}(\{MSA_{k+1}=w_l \} \cap \left\{\tau_{k+1}=w_{2k+2}\right\} |X_{(k+1)}=w_l)$, for $l \in \{k+1, \dots, 2k+1 \}$. 
% If $w_l$ is blocked and $X_{(k+1)}=w_l$, then there are more than $k+1$ elements greater than $w_{2k+2}$ in the set $X$, which is incompatible with having $\tau_{k+1}=w_{2k+2}$. Therefore,
% \[\m{P}(\{MSA_{k+1}=w_l \} \cap \left\{\tau_{k+1}=w_{2k+2}\right\} |X_{(k+1)}=w_l)= 0\] 
% if $w_l$ is blocked. 

% If $w_l$ is not blocked, by Lemma~\ref{ineq:MSA:2} we have that $\m{P}(\{MSA_{k+1}=w_l \} \cap \left\{\tau_{k+1}=w_{2k+2}\right\} |X_{(k+1)}=w_l)$ is at least $2^{-2k+l-2}$ if $w_l$  is not ill-paired, and at least $  2^{-2k+l-1}$ if $w_l$ is ill-paired. 
%Putting everything together, 
By Lemma~\ref{ineq:MSA:2} and by the definition of $\delta_l$, 
we have that for each $l \in \{k+1, \dots, 2k+1\}$, 
\[
\m{P}(\{MSA_{k+1}=w_l \} \cap \left\{\tau_{k+1}=w_{2k+2}\right\} |X_{(k+1)}=w_l) \geq \frac{\delta_l}{2}, 
\]
and the result follows. 

\end{proof}

\subsubsection*{Proof of Lemma \ref{lem:bound:msak+1:lem2}}

First, we decompose the left-hand-side term in Lemma \ref{lem:bound:msak+1:lem1} as follows: 
\begin{eqnarray*}
 \EE(MSA_{k+1}1_{\tau_{k+1} \leq w_{2k+3}})&=& \sum_{l=2k+3}^{\xi_{k+1}} \EE(MSA_{k+1}| \tau_{k+1}=w_l) \PP(\tau_{k+1}=w_l) \\
\end{eqnarray*}
Since each $w_l$, for $l \geq 1$,
is equally likely to be in $S$ or in $X$, the law of $\tau_{k+1}$ is identical to the law of $X_{(k+1)}$. We deduce that for all $l \geq 1$, $\PP(\tau_{k+1}=w_l)=\PP(X_{(k+1)}=w_l)$. Secondly, when $\tau_{k+1}=w_l$, there are $l-k-1 \geq k+1$ elements above $w_l$ that are in $X$. Hence, $MSA_{k+1}$ will pick one of them, and we deduce that 
$\EE(MSA_{k+1}| \tau_{k+1}=w_l) \geq w_{l-1}$. These two observations give
\begin{eqnarray} 
%\m{E}(MSA_{k+1} 1_{\tau_{k+1} \leq w_{2k+3}}) &=& \sum_{i=k+1
 \EE(MSA_{k+1}1_{\tau_{k+1} \leq w_{2k+3}}) 
 &\geq&
 \sum_{l=2k+3}^{\xi_{k+1}} w_{l-1} \PP(X_{k+1}=w_l) \nonumber
 \\
&=& \sum_{l=2(k+1)}^{\xi_{k+1}-1} w_l \PP(X_{(k+1)}=w_{l+1}) \label{eq:proof:msak+1}
\end{eqnarray}
One of the main differences between the above inequality and the one we want to prove in Lemma \ref{lem:bound:msak+1:lem2} is that the term inside the sum is $\PP(X_{(k+1)}=w_{l+1})$ instead of $\PP(X_{(k+1)}=w_{l})$. In the sequel, we relate these two quantities. First, we compute $\PP(X_{(k+1)}=w_{l})$. \\
%The first result (Lemma~\ref{lem:prophet}) describes the distribution of the $(k+1)$-th largest value in $X$. \\

\begin{lemma}\label{lem:prophet}  The probability distribution of $X_{(k+1)}$ %conditional on $\xi_j$ for  $j \in \{1, \dots, k+1\}$ 
is given by
        \[    
  \PP(X_{(k+1)}=w_l) =  \begin{cases}
  \displaystyle\frac{\binom{l-1-2j}{k-j}}{2^{l-2j}}  &\quad\text{if 
       } l \in \{\xi_j+1, \dots, \xi_{j+1}-1 \}, \text{ for }  j \in \{0, \dots, k\} \\
   \noalign{\vskip9pt}
 \displaystyle\frac{\binom{\xi_j-2j}{k+1-j}}{2^{\xi_j-2j+1}}  &\quad\text{if 
       } l= \xi_j, \text{ for }  j \in \{1, \dots, k+1\}. \\
     \end{cases}
\]
\end{lemma}

\begin{proof}[Proof of Lemma~\ref{lem:prophet}]

We divide the proof into two cases, depending on whether $l= {\xi_j}$ for some $j \in \{0, \dots, k+1\}$ or not.  \\

 \noindent{\bf Case 1:} Suppose that $l \in \{\xi_j+1, \dots, \xi_{j+1}-1\}$ for some $j \in \{0, \dots, k\}$. Note that $ X_{(k+1)}=w_l$ if and only if $w_l \in X$ and there are exactly $k$ values in $X$ that are larger than $w_l$. 
 
 Since $l \in \{\xi_j+1, \dots, \xi_{j+1}-1\}$, we have, conditioned on $w_l \in X$, that there are $j+1$ values in $X$ and $j$ in $S$ with probability 1. Therefore, the probability that exactly $k$ values in $X$ are among the $l-1$ largest values is given by 
   \[
  \binom{l-2j-1}{k-j} \frac{1}{2^{l-2j-1}}.
  \]

 On the other hand, $\PP(w_l \in X)=1/2$, and thus we conclude that in this case, 
    \begin{equation}  \label{eqn:prophet:prob1}
       \PP(X_{(k+1)}=w_l)=  \binom{l-2j-1}{k-j} \frac{1}{2^{l-2j}}.    
    \end{equation} 
    \\

\noindent{\bf Case 2:} Suppose that $l= \xi_j$ for some $j \in \{1, \dots, k+1\}$. The analysis in this case is similar to that of Case 1. However, the probability of having exactly $k$ values in $X$ greater than $w_l$ is now given by

  \[
  \binom{l-2j}{k-(j-1)} \frac{1}{2^{l-2j}}.
  \]
  
 In effect, conditioned on $w_l \in X$, there are $j-1$ values in $X$ greater than $w_l$ and $j$ values greater than $w_l$ in $S$, with probability one. Thus, we need to compute the probability that exactly $k-(j-1)$ additional values in $X$ come from the $l-2j$ remaining elements. This probability is given by the expression above.

Therefore, in this case, 

\begin{equation}\label{eqn:prophet:prob2}
     \PP(X_{(k+1)}=w_l)=   \binom{l-2j}{k-(j-1)} \frac{1}{2^{l-2j+1}}
\end{equation}
Combining \eqref{eqn:prophet:prob1} and \eqref{eqn:prophet:prob2}, we obtain the desired result.
\end{proof}

We now use the previous lemma to lower bound $\PP(X_{(k+1)}=w_{l+1})$ in terms of $\PP(X_{(k+1)}=w_l)$. As suggested by the expression in Lemma \ref{lem:prophet}, we will need to distinguish between the cases where $l$ and $l+1$ are some $\xi_j$ or not. \\

%\begin{proof}
\begin{lemma} \label{lem:ll+1}
 %$l \in \left\{2k+2,\dots,\xi_{k+1}-1\right\}$.
\begin{enumerate}[label=\alph*)]
    \item \label{lem:ll+1a}
    Let $j \in \left\{0,\dots,k\right\}$ and $l \in \{\xi_j+1, \dots, \xi_{j+1}-1 \}$. 
     \[    
   \PP(X_{(k+1)}=w_{l+1}) \geq  \begin{cases}
  \displaystyle \frac{1}{2} \PP(X_{(k+1)}=w_{l})  &\quad\text{if 
       } l+1 \in \{\xi_j+1, \dots, \xi_{j+1}-1 \} \\
   \noalign{\vskip9pt}
  \PP(X_{(k+1)}=w_{l})  &\quad\text{if 
       } l+1= \xi_{j+1}. \\
     \end{cases}
\]
    \item \label{lem:ll+1b}
    Assume that $2k+2=\xi_{j}$, for some $j \in \left\{1,\dots,k\right\}$. 
    Then
    \[    
   \PP(X_{(k+1)}=w_{2k+3}) \geq  \begin{cases}
  \displaystyle \frac{1}{2} \PP(X_{(k+1)}=w_{2k+2})  &\quad\text{if 
       } 2k+3 \neq \xi_{j+1} \\
   \noalign{\vskip9pt}
  \PP(X_{(k+1)}=w_{2k+2})  &\quad\text{if 
       } 2k+3= \xi_{j+1}. \\
     \end{cases}
\]
    %and $2k+3=\xi_{j+1}$ for some $j$, then
    % $\PP(X_{(k+1)}=w_{2k+3}) = \PP(X_{(k+1)}=w_{2k+2})$
    % \item
     %  If $2k+2=\xi_{j}$ for some $j<k+1$ and $2k+3 \neq \xi_{j+1}$, then
    % $\PP(X_{(k+1)}=w_{2k+3}) \geq \frac{1}{2} \PP(X_{(k+1)}=w_{2k+2})$
     \end{enumerate}
     \end{lemma}
     \begin{proof}[Proof of Lemma~\ref{lem:ll+1}]
     \begin{enumerate}[label=\alph*)]
     \item[$a)$]
     Assume $l+1 \in \{\xi_j+1, \dots, \xi_{j+1}-1 \}$. We have
     \begin{eqnarray*}
      \PP(X_{(k+1)}=w_{l+1}) &=& \frac{\binom{l-2j}{k-j}}{2^{l+1-2j}}
      \\
      &\geq& \frac{1}{2} \cdot \frac{\binom{l-1-2j}{k-j}}{2^{l-2j}}
      \\
      &=& \frac{1}{2} \cdot \PP(X_{(k+1)}=w_{l})
     \end{eqnarray*}
   Assume that $l+1=\xi_{j+1}$. We have
      \begin{eqnarray*}
      \PP(X_{(k+1)}=w_{l+1}) &=& \frac{\binom{\xi_{j+1}-2j-2}{k-j}}{2^{\xi_{j+1}-2j-1}}
      \\
      &=& \PP(X_{(k+1)}=w_{l})  
     \end{eqnarray*}
     \item[$b)$]
     Assume that $2k+3 \neq \xi_{j+1}$. 
     We have
     \begin{eqnarray*}
     \PP(X_{(k+1)}=w_{2k+3})&=& \frac{\binom{2k+2-2 j}{k-j}}{2^{2k+3-2j}} 
     \\
     &=& \left(\frac{k+1-j}{k+2-j} \right) \frac{\binom{2k+2-2 j}{k+1-j}}{2^{2k+3-2j}} 
     \\
     &\geq & \frac{1}{2} \PP(X_{(k+1)}=w_{2k+2})
     \end{eqnarray*}
     Assume that $2k+3 = \xi_{j+1}$.
     We have
     \begin{eqnarray*}
     \PP(X_{(k+1)}=w_{2k+3})&=& \frac{\binom{2k+1-2 j}{k-j}}{2^{2k+2-2j}} 
     \\
     &=&  \frac{\frac{1}{2} \binom{2k+2-2 j}{k+1-j}}{2^{2k+2-2j}}
     \\
     &=& \PP(X_{(k+1)}=w_{2k+2})
     \end{eqnarray*}
     \end{enumerate}
     \end{proof}
     We are now ready to prove Lemma \ref{lem:bound:msak+1:lem2}. 
     \begin{proof}[Proof of Lemma~\ref{lem:bound:msak+1:lem2}]
     By inequality \eqref{eq:proof:msak+1}, it is enough to prove that 
     \begin{equation*}
     \sum_{l=2(k+1)}^{\xi_{k+1}-1} w_l \PP(X_{(k+1)}=w_{l+1}) 
     \geq \frac{1}{2} \sum_{l=2k+2}^{\xi_{k+1}} w_{l} \PP(X_{(k+1)}=w_{l}).
     \end{equation*}
     \noindent{\bf Case 1.} $2k+3=\xi_{j}$, for some $j \in \left\{1,\dots,k\right\}$.
     
     By Lemma \ref{lem:ll+1} \ref{lem:ll+1b}, we have
     \begin{equation} \label{eqn:ll+1}
      \PP(X_{(k+1)}=w_{2k+3}) \geq  \frac{1}{2} \PP(X_{(k+1)}=w_{2k+2})+\frac{1}{2} \PP(X_{(k+1)}=w_{2k+3})
     \end{equation}
     We deduce that
     \begin{eqnarray*}
     \sum_{l=2(k+1)}^{\xi_{k+1}-1} w_l \PP(X_{(k+1)}=w_{l+1}) &=& w_{2k+2} \PP(X_{(k+1)}=w_{2k+3})+\sum_{l=2k+3}^{\xi_{k+1}-1} w_l \PP(X_{(k+1)}=w_{l+1}) 
     \\
     &\geq & \frac{1}{2} w_{2k+2} \PP(X_{(k+1)}=w_{2k+2})+\frac{1}{2} w_{2k+3} \PP(X_{(k+1)}=w_{2k+3})
     \\
     &+&
     \sum_{l=2k+3}^{\xi_{k+1}-1} w_{l+1} \PP(X_{(k+1)}=w_{l+1}) 
     \\
     &\geq& \frac{1}{2} \sum_{l=2k+2}^{\xi_{k+1}} w_{l} \PP(X_{(k+1)}=w_{l}),
     \end{eqnarray*}
     where in the second-to-last inequality, we used \eqref{eqn:ll+1} and the fact that $w_{2k+2} \geq w_{2k+3}$ and $w_l \geq w_{l+1}$. \\
     
          \noindent{\bf Case 2.} $2k+3 \in \left\{\xi_{j}+1,\dots,\xi_{j+1}-1 \right\}$ for some $j \in \left\{0,\dots,k\right\}$.

The sum $\sum_{l=2k+2}^{\xi_{k+1}-1} w_l \m{P}(X_{(k+1)}=w_{l+1})$ can be decomposed as 
\[
\sum_{l=2k+2}^{\xi_{j+1}-2} w_l \m{P}(X_{(k+1)}=w_{l+1})+  w_{\xi_{j+1}-1} \PP(X_{(k+1)}=w_{\xi_{j+1}})
    +\sum_{l=\xi_{j+1}}^{\xi_{k+1}-1} w_{l} \PP(X_{(k+1)}=w_{l+1}).
\]

By using  Lemma \ref{lem:ll+1} and the fact that $w_l \geq w_{l+1}$, we can lower bound the expression above by 
\[
\frac{1}{2} \sum_{l=2k+2}^{\xi_{j+1}-2} w_l \m{P}(X_{(k+1)}=w_{l})+ w_{\xi_{j+1}-1} \PP(X_{(k+1)}=w_{\xi_{j+1}-1})
    +\sum_{l=\xi_{j+1}}^{\xi_{k+1}-1} w_{l+1} \PP(X_{(k+1)}=w_{l+1}),
\]
which is at least
\[
\frac{1}{2}  \sum_{l=2k+2}^{\xi_{k+1}} w_l \m{P}(X_{(k+1)}=w_{l}),
\]
as we wanted to see. 

\end{proof}

\subsubsection{Proof of Proposition \ref{lem:bound:msai}}

Before proving  Proposition \ref{lem:bound:msai}, we introduce one technical lemma that gives a lower bound for the probability of $\overline{MSA}_i$ picking a value $w_l$ conditional on $w_l$ being the $(k+1)$-largest value in the set $X$. \\ 
\begin{lemma} \label{ineq:MSA}
Let $l \in \{1, \dots, 2k+1\}$.

\begin{enumerate}
    \item[a)] If $w_l$ is not blocked and not ill-paired, for all $i \in \left\{l-k, \dots, k \right\}$ it holds 
    \[
    \PP(\overline{MSA}_i=w_l | X_{(k+1)}=w_l) \geq 2^{-k-i+l-1} .
    \]
    \item[b)] If $w_l$ is blocked, and that either $w_l$ is not ill-paired, or it is ill-paired and $m_l<p$, we have

     \[    
   \PP(\overline{MSA}_i=w_l | X_{(k+1)}=w_l) \geq  \begin{cases}
  \displaystyle 2^{-k-i+l-1}   &\quad\text{if 
       } i \in \{l-k, \dots, m_l-k-3\}, \\
   \noalign{\vskip9pt}
 \displaystyle2^{-k-i+l}   &\quad\text{if 
       } i = m_l-k-2. \\
     \end{cases}
\]
  \item[c)] Assume $w_l$ is not blocked and ill-paired. Then, 
     \[    
   \PP(\overline{MSA}_i=w_l | X_{(k+1)}=w_l) \geq  \begin{cases}
  \displaystyle 2^{-k-i+l-1}   &\quad\text{if 
       } i \in \{l-k, \dots, p-k-2\}, \\
   \noalign{\vskip9pt}
 \displaystyle2^{-k-i+l}   &\quad\text{if 
       } i \in \{p-k, \dots, k\}. \\
     \end{cases}
\]
\item[d)] Assume $w_l$ is blocked and ill-paired, and that $m_l>p$. Then,
     \[    
   \PP(\overline{MSA}_i=w_l | X_{(k+1)}=w_l) \geq  \begin{cases}
  \displaystyle 2^{-k-i+l-1}   &\quad\text{if 
       } i \in \{l-k, \dots, p-k-2\}, \\
   \noalign{\vskip9pt}
 \displaystyle2^{-k-i+l}   &\quad\text{if 
       } i \in \{p-k, \dots,m_l-k-3\},\\
       \noalign{\vskip9pt}
 \displaystyle2^{-k-i+l+1}   &\quad\text{if 
       } i = m_l-k-2.
     \end{cases}
\]
\end{enumerate}
\end{lemma}
\begin{proof}[Proof of Lemma~\ref{ineq:MSA}]
Let $l \in \{1, \dots, 2k+1\}$. 
\begin{enumerate}
\item[$a$)] 
The proof is very similar to the one of Lemma \ref{ineq:MSA:2} \ref{ineq:MSA:2a}, up to replacing $k+1$ by $i$. For sake of completeness, we draw the main lines. 
Take $i \in \{l-k, \dots, k\}$. We want to analyze
 $\m{P}(\overline{MSA}_i=w_l| X_{(k+1)}=w_l) $.
 First, note that since $l$ is not ill-paired, $w_l$ is not paired with $w_{k+i+1}$. Then, 
 \begin{eqnarray*}
  \m{P}(\overline{MSA}_i=w_l| X_{(k+1)}=w_l) & \geq &
      \m{P}(\left\{\overline{MSA}_i=w_l\right\} \cap \left\{\tau_i=w_{k+i+1}\right\} |X_{(k+1)}=w_l) 
      \\
      &=& 
       \m{P}(\tau_i=w_{k+i+1}|X_{(k+1)}=w_l),
  \end{eqnarray*}
  where the equality stems from the fact that, when $X_{(k+1)}=w_l$ and $\tau_{i}=w_{k+i+1}$, all the elements in $\left\{w_{l+1},\dots,w_{k+i}\right\}$ must be in $S$, and then $\overline{MSA}_i$ picks $w_l$, because it is not paired with the threshold $\tau_i$. 
  
  Therefore, it is enough to prove that $\m{P}(\tau_i=w_{k+i+1}|X_{(k+1)}=w_l) \geq 2^{-k-i+l-1}$.  Given $X_{(k+1)}=w_l$, in order for $\tau_i=w_{k+i+1}$ to hold, it is necessary and sufficient that all the elements in $\left\{w_{l+1},\dots,w_{k+i+1}\right\}$ belong to $S$. This event occurs with probability greater than $2^{-k-i+l-1}$, by a similar computation as in the proof of Lemma \ref{ineq:MSA:2} \ref{ineq:MSA:2a}. \\

\item[$b)$] Take 
       $ i \in \{l-k, \dots, m_l-k-3\}$. In this case, $ k+i+1<m_l$, hence the pair that blocks $w_l$ is smaller than $w_{k+i+1}$. Moreover, since either $w_l$ is not ill-paired or $m_l<p$, $w_l$ is not paired with $w_{k+i+1}$. We can therefore replicate the same computations as in $a)$, and thus  obtain the claimed inequality.  
       
       If $i=m_l-k-2$, we can replicate the same computations as in $a)$ too, which yields: 
       \begin{eqnarray*}
      \m{P}(\overline{MSA}_i=w_l \cap \left\{\tau_i=w_{k+i+1} \right\} |X_{(k+1)}=w_l) 
      &\geq & 2^{-k-i+l-1}. 
  \end{eqnarray*}
To obtain the desired lower bound, we will consider in addition the case where $\tau_i=w_{m_l}$. 
    Indeed, whenever $X_{(k+1)}=w_l$ and $\tau_i=w_{m_l}$, the only element in 
  $X$ that is below $w_l$ and above the threshold $\tau_i$ is $w_{m_l}$'s pair, namely $w_{m'}$. By definition of $\overline{MSA}_i$, such an element is not selected, and therefore $\overline{MSA}_i$ selects $w_l$. We deduce that
   \begin{eqnarray} \nonumber
      \m{P}(\left\{\overline{MSA}_i=w_l\right\} \cap \left\{\tau_i=w_{m_l} \right\} |X_{(k+1)}=w_l) 
      &=& 
       \m{P}(\tau_i=w_{m_l}|X_{(k+1)}=w_l).
  \end{eqnarray}
    Knowing $X_{(k+1)}=w_l$, in order to get $\tau_i=w_{m_l}=w_{k+i+2}$, it is necessary and sufficient that all the elements in $\left\{w_{l+1},\dots,w_{k+i+2}\right\}\setminus \left\{m'\right\}$ belong to $S$, which happens with probability higher than $2^{-k-i+l-1}$, 
    by a similar computation as in the proof of Lemma \ref{ineq:MSA:2} \ref{ineq:MSA:2a}. Then, we have
  \begin{eqnarray*}
      \m{P}(\overline{MSA}_i=w_l \cap \left\{\tau_i=w_{k+i+2} \right\} |X_{(k+1)}=w_l) 
      &\geq & 2^{-k-i+l-1}. 
  \end{eqnarray*}
 
We conclude that  
\begin{eqnarray*}
      \m{P}(\overline{MSA}_i=w_l |X_{k+1}=w_l)  &\geq&  \m{P}(\overline{MSA}_i=w_l \cap \left\{\tau_i=w_{k+i+2} \right\} |X_{(k+1)}=w_l) \\
      &+&   \m{P}(\overline{MSA}_i=w_l \cap \left\{\tau_i=w_{k+i+1} \right\} |X_{(k+1)}=w_l)   \\ 
      &\geq & 2^{-k-i+l}, 
  \end{eqnarray*}
which is the desired result. \\
\item[$c)$] If $i \in \{l-k, \dots, p-k-2\}$, the argument proceeds as in part $a)$.

Take $i \in \left\{p-k,\dots,k\right\}$. In this case, $k+i+1>p$, and then $w_l$ is not paired with $w_{k+i+1}$. We therefore have
     \begin{eqnarray*} 
  \m{P}(\overline{MSA}_i=w_l| X_{(k+1)}=w_l) & \geq &
      \m{P}(\overline{MSA}_i=w_l \cap \left\{\tau_i=w_{k+i+1}\right\} |X_{(k+1)}=w_l) 
      \\
      &=& 
       \m{P}(\tau_i=w_{k+i+1}|X_{(k+1)}=w_l),
  \end{eqnarray*}

  and to obtain the result it is enough to show that $$ \m{P}(\tau_i=w_{k+i+1}|X_{(k+1)}=w_l) \geq 2^{-k-i+l}.$$
  
   Given that $X_{(k+1)}=w_l$, we know that $w_p$ is in $S$, since $(w_p,w_l)$ is a pair. Then, in order to get $\tau_i=w_{k+i+1}$, it is necessary and sufficient that all the elements in $\left\{w_{l+1},\dots,w_{k+i+1}\right\} \setminus \left\{p\right\}$ belong to $S$, which happens with probability at least $2^{-k-i+l}$, by a similar computation as in the proof of Lemma \ref{ineq:MSA:2} \ref{ineq:MSA:2a}. We deduce that 
   \begin{eqnarray}\label{ineq:MSA_i}  \nonumber
      \m{P}(\overline{MSA}_i=w_l| X_{(k+1)}=w_l) &\ge& 
 2^{-k-i+l},
  \end{eqnarray}
which is what we wanted to show. \\ 
  \item[$d)$]
  The first two cases can be proved as in $a)$ and $c)$. Let $i=m_l-k-2$; that is $k+i+2=m_l$. We call $w_{m'}$ the pair of $w_{m_l}$. 
  
  First, we can replicate the computations of Case $c)$, and obtain:
  \begin{eqnarray*}
      \m{P}( \{\overline{MSA}_i=w_l \} \cap \left\{\tau_i=w_{k+i+1} \right\} |X_{(k+1)}=w_l) 
      &\geq & 2^{-k-i+l}. 
  \end{eqnarray*}
As in Case $b)$, in order to obtain the claimed bound of the lemma, we need to consider the event $\left\{\tau_i=w_{k+i+2} \right\}$. We have
   \begin{eqnarray*} \nonumber
      \m{P}(\overline{MSA}_i=w_l| X_{(k+1)}=w_l) & \geq &   \m{P}( \{\overline{MSA}_i=w_l \} \cap \left\{\tau_i=w_{k+i+2} \right\} |X_{(k+1)}=w_l) \\
      &=& 
       \m{P}(\tau_i=w_{k+i+2}|X_{(k+1)}=w_l),
  \end{eqnarray*}
   where the equality stems from the fact that, when $X_{(k+1)}=w_l$ and $\tau_{i}=w_{k+i+2}$, all the elements in $\left\{w_{l+1},\dots,w_{k+i+1}\right\} \setminus \left\{w_{m'}\right\}$ must be in $S$ and then $\overline{MSA}_i$ picks $w_l$, because $w_l$ is not paired with $w_{k+i+2}$.
    Since $w_l$ and $w_p$ are paired, given that $X_{(k+1)}=w_l$, we have that $w_p$ lies in $S$. Moreover, since $w_{k+i+2}$ and $w_{m'}$ are paired, if $w_{k+i+2}$ lies in $S$, then $w_{m'}$ lies in $X$. Hence, in order to get $\tau_i=w_{k+i+2}$, it is necessary and sufficient that all the elements in $\left\{w_{l+1},\dots,w_{k+i+2}\right\}\setminus \left\{w_{m'},w_p\right\}$ belong to $S$, which happens with probability at least $2^{-k-i+l}$, by a similar computation as in the proof of Lemma \ref{ineq:MSA:2} \ref{ineq:MSA:2a}. 
    We deduce that 
  \begin{eqnarray*}
      \m{P}(\{\overline{MSA}_i=w_l \} \cap \left\{\tau_i=w_{k+i+2} \right\} |X_{(k+1)}=w_l) 
      &\geq & 2^{-k-i+l}. 
  \end{eqnarray*}
We conclude that  
\begin{eqnarray*}
      \m{P}(\overline{MSA}_i=w_l |X_{(k+1)}=w_l)  &\geq&  \m{P}(\{\overline{MSA}_i=w_l \} \cap \left\{\tau_i=w_{k+i+2} \right\} |X_{(k+1)}=w_l) \\
      &+&   \m{P}(\{\overline{MSA}_i=w_l \} \cap \left\{\tau_i=w_{k+i+1} \right\} |X_{(k+1)}=w_l)   \\ 
      &\geq & 2^{-k-i+l+1}.
  \end{eqnarray*}
\end{enumerate}
\end{proof}

\begin{proof}[Proof of Proposition~\ref{lem:bound:msai}]
Note that
 \begin{eqnarray*}
\sum_{i=1}^{k} \m{E}(\overline{MSA}_i) & \geq&  \sum_{i=1}^{k} \sum_{l=1}^{2k+1} \m{P}(\overline{MSA}_{i}=w_l | X_{(k+1)}=w_l) w_l \PP(X_{(k+1)}=w_l) \\
&=& \sum_{l=1}^{2k+1}  w_l \PP(X_{(k+1)}=w_l) \sum_{i=1}^{k} \m{P}(\overline{MSA}_{i}=w_l | X_{(k+1)}=w_l) .
\end{eqnarray*}

In the remainder of the proof we show that for each $l \in \{1, \dots, 2k+1\}$, 
\[
\sum_{i=1}^{k} \m{P}(\overline{MSA}_{i}=w_l | X_{(k+1)}=w_l) = 1- \delta_l
, \]
 where we recall that

\[ \delta_l=  \left\{
\begin{array}{ll}
      2^{-2k+l-1} & \text{ if } w_l \text{ is not blocked and not ill-paired} \\
     2^{-2k+l} & \text{ if } w_l \text{ is not blocked and ill-paired} \\
      0 & \text{ otherwise.}\\
\end{array} 
\right. \] \\

\noindent{\bf Case 1.} $w_l$  is not blocked and not ill-paired.

In this case, by Lemma~\ref{ineq:MSA}

\[
\sum_{i=1}^{k} \m{P}(\overline{MSA}_{i}=w_l | X_{(k+1)}=w_l) = \sum_{i=l-k}^{k} 2^{-k-i+l-1}=1 - 2^{-2k + l -1}. 
\]
\\

\noindent{\bf Case 2.} $w_l$  is not blocked and ill-paired.
\[
\sum_{i=1}^{k} \m{P}(\overline{MSA}_{i}=w_l | X_{(k+1)}=w_l) = \sum_{i=l-k}^{p-k-2} 2^{-k-i+l-1}+ \sum_{i=p-k}^{k} 2^{-k-i+l} =1 - 2^{-2k + l}.
\]
\\

\noindent{\bf Case 3.} $w_l$ is blocked, and that either $w_l$ is not ill-paired, or it is ill-paired and $m_l<p$. 

In this case, 
\[
\sum_{i=1}^{k} \m{P}(\overline{MSA}_{i}=w_l | X_{(k+1)}=w_l) = \sum_{i=l-k}^{m_l-k-3} 2^{-k-i+l-1}+  2^{l-m_l+2}=1.
\]
\\

\noindent{\bf Case 4.} $w_l$ is blocked and ill-paired, and $m_l>p$.
In this case, 
\[
\sum_{i=1}^{k} \m{P}(\overline{MSA}_{i}=w_l | X_{(k+1)}=w_l) = \sum_{i=l-k}^{p-k-2} 2^{-k-i+l-1}+ \sum_{i=p-k}^{m_l-k-3} 2^{-k-i+l}+ 2^{-m_l+l+3} = 1.
\]

Putting everything together, we obtain the desired result.
\end{proof}

\subsection{Proof of Theorem \ref{thm:gen1'}} \label{sec:thm1:1'}

In order to prove Theorem~\ref{thm:gen1'}, we use algorithms $MSA_1, \dots, MSA_{k+1}$ defined in Section~\ref{sec:main}.
%the $k+1$ algorithms analogously to the previous section, but this time all of them will be a single-threshold algorithms. More precisely, in algorithm $MSA_i$, for $i \in \{1, \dots, k+1\},$ the strategy for the gambler consists of taking as threshold $\tau$ the $i$-th largest value in the sample set, and picking the first value $x_j$ above $\tau$. 
Then, we define the randomized strategy $MSA_{RAND}$ as follows:
(1) before the game starts, select a number $I$ in 
$\left\{1,\dots,k+1\right\}$ uniformly at random, that is, $I=i$ with probability $1/(k+1)$. (2) Play $MSA_I$. \\

Note that, unlike the randomized algorithm we used to prove Theorem 1, here $MSA_{RAND}$ does not need access to the identity of the arriving variables. In the next proposition, we prove that $MSA_{RAND}$ has a competitive ratio of at least $\frac{1}{2(k+1)}$. This directly implies Theorem~\ref{thm:gen1'}. Indeed, $MSA_{RAND}$ is a randomization over single-threshold algorithms. By linearity of expectation, there exists a single-threshold strategy in the support of $MSA_{RAND}$ that performs as well as $MSA_{RAND}$. 

\begin{proposition} \label{prop:NI}
    The strategy $MSA_{RAND}$ has a competitive ratio $\frac{1}{2(k+1)}$.
\end{proposition}
\begin{proof}
We want to prove that $\EE(MSA_{RAND}) \geq \frac{1}{2(k+1)} \EE(X_{(k+1)}).$ First, note that 
\begin{eqnarray*}
   (k+1) \EE(MSA_{RAND}) &=&  
    \sum_{i=1}^{k+1} \EE(MSA_i) \\
    &=& \sum_{i=1}^{k} \EE(MSA_i | X_{(k+1)}=w_{k+i}) \PP(X_{(k+1)}=w_{k+i}) + \EE(MSA_{k+1}).
\end{eqnarray*}
Now, using that for each $i \in \{1, \dots, k\} $
\[
\EE(MSA_i | X_{(k+1)}=w_{k+i}) \geq w_{k+i} \PP(MSA_i=w_{k+i}|X_{(k+1)}=w_{k+i}),
\]
we have that $(k+1) \EE(MSA_{RAND})$ is at least 
\[
\sum_{i=1}^{k} w_{k+i} \PP(MSA_i=w_{k+i}|X_{(k+1)}=w_{k+i})\PP(X_{(k+1)}=w_{k+i}) + \EE(MSA_{k+1}).
\]
In the remainder of the proof we bound $\PP(MSA_i=w_{k+i}|X_{(k+1)}=w_{k+i})$ for each $i \in \{1, \dots, k\}$, and $\EE(MSA_{k+1})$.\\

\noindent{\bf Step 1.} For each $i \in \{1, \dots, k\}$, $$\PP(MSA_i=w_{k+i}|X_{(k+1)}=w_{k+i}) \geq 1/2.$$

On one hand, 
 \begin{eqnarray*}
  \m{P}(MSA_i=w_{k+i}| X_{(k+1)}=w_{k+i}) & \geq &
      \m{P}(\left\{MSA_i=w_{k+i}\right\} \cap \left\{\tau_i=w_{k+i+1}\right\} |X_{(k+1)}=w_{k+i}) 
      \\
      &=& 
       \m{P}(\tau_i=w_{k+i+1}|X_{(k+1)}=w_{k+i}),
  \end{eqnarray*}
  where the equality stems from the fact that, when $X_{(k+1)}=w_{k+i}$ and $\tau_{i}=w_{k+i+1}$,  $MSA_i$ picks $w_{k+i}$. 

On the other hand, it is easy to see that  $\m{P}(\tau_i=w_{k+i+1}|X_{(k+1)}=w_l) $ is  $1$ if $w_{k+i}$ and $w_{k+i+1}$ are paired, and $1/2$, otherwise. Therefore, we obtain
$\m{P}(\tau_i=w_{k+i+1}|X_{(k+1)}=w_l) \geq 2^{-1}$. The first step is completed.

\noindent{\bf Step 2.} $\EE(MSA_{k+1}) \geq 1/2 \sum_{l=2k+1}^{\xi_{k+1}} w_l \PP(X_{(k+1)}=w_l). $ \\

If $2k+2=\xi_{k+1}$, the elements of $\left\{w_1,\dots,w_{2k+2}\right\}$ form $k+1$ pairs. Hence, if $w_{2k+2} \in S$, which happens with probability $1/2$, there are exactly $k+1$ elements larger than $w_{2k+2}$ that are in $X$. In that case, $MSA_{k+1}$ picks $X_{(k+1)}$. It follows that 
\begin{equation*}
\m{E}(MSA_{k+1}) \geq \frac{1}{2} \m{E}(X_{(k+1)}) \geq\frac{1}{2}\sum_{l=2k+1}^{\xi_{k+1}} w_l \PP(X_{(k+1)}=w_l).
\end{equation*}

If $2k+2 \neq \xi_{k+1}$, 
\[
\EE(MSA_{k+1}) = \EE(MSA_{k+1} 1_{\tau_{k+1}=w_{2k+2}}) + \EE(MSA_{k+1} 1_{\tau_{k+1} \leq w_{2k+3}}).
\]
By Lemma~\ref{lem:bound:msak+1:lem1},
\[
\EE(MSA_{k+1} 1_{\tau_{k+1}=w_{2k+2}}) \ge 
\PP(X_{(k+1)}=w_{2k+1}) w_{2k+1} \frac{\delta_{2k+1}}{2}, 
\] 
where $\delta_{2k+1}$ is equal to 1 since by definition $w_{2k+1}$ cannot be blocked nor ill-paired. On the other hand, by Lemma~\ref{lem:bound:msak+1:lem2}, the second term is lower bounded by \[
 \frac{1}{2} \sum_{l=2(k+1)}^{\xi_{k+1}} \PP(X_{(k+1)}=w_l) w_l.
\]
Putting all together, we obtain Step 2.

Combining Step 1 and Step 2, we conclude that 
\begin{eqnarray*}
   (k+1) \EE(MSA_{RAND})  & \geq & \frac{1}{2} \sum_{l=k+1}^{\xi_{k+1}} w_l \PP(X_{(k+1)}=w_l)= \EE(X_{(k+1)}),
\end{eqnarray*}
and the proof is completed. 
\end{proof}

\section{Upper bound on competitive ratio}\label{subsec:hard_instance} 

In this section, we provide two tightness results. The first one, Theorem~\ref{thm:gen2}, establishes the tightness of the competitive ratio result for FI $\RPI$ presented in Section~\ref{sec:main}, by providing a parameterized hard distribution for FI $\RPI$ showing that no algorithm can have a competitive ratio of $\gamma$, with $\gamma> 1/(k+2)+\beta$ for any $\beta>0$. This leads to a similar result for the lower-information model NI. \\ % and NI. \\
The second result, Proposition~\ref{prop:tight:NI}, shows that the strategy $MSA_{RAND}$ cannot guarantee a competitive ratio better than $1/(2k+2)$ in the NI model. \\

 \begin{theorem}
      \label{thm:gen2}
    For each $\beta>0$, there exists an instance with $2(k+1)$ variables such that no algorithm has a competitive ratio larger than $1/(k+2)+\beta$, regardless of the information model in $\RPI$.

  \end{theorem}

\begin{proof}[Proof of Theorem~\ref{thm:gen2}] Let $0<\varepsilon < 1/(k+1)$. Consider the following $2(k+1)$ random variables: $X_i =\frac{1}{\varepsilon^{i-1}\binom{k+1}{i-1}}$ for $i\in \{ 1,\ldots,k+1\}$ 
and for $i\in \{ k+2,\ldots, 2(k+1)  \} $
\[
X_i=\begin{cases}
    0 & \text{w.p. }1-\varepsilon,\\
    1/\varepsilon^{k+1} & \text{w.p. }\varepsilon.
\end{cases}
\]
 
We now prove that for the instance with these $2(k+1)$ random variables, no strategy of the gambler can attain a competitive ratio larger than $1/(k+2) +\mathcal{O}(\varepsilon)$.

Note that, as $\varepsilon <(k+1)^{-1}$, it holds $X_i < X_{i+1}$ for $1 \leq i \leq k$. That is, the deterministic variables arrive in increasing order. Indeed, for $i \in \{1, \dots, k\}$, $X_i > X_{i+1}$ if and only if $
    \varepsilon < \frac{i}{k-i+2} $. As $
    \frac{i}{k-i+2} $ is increasing in $i$, it is enough to have $\varepsilon < (k+1)^{-1}$.

Let us compute $\m{E}(X_{(k+1)})$. To this end, note that the $(k+1)$-th largest variable corresponds to $X_i$ with $i\leq k+1$ if and only if exactly $i-1$ variables take the value $1/\varepsilon^{k+1}$ (because the deterministic variables arrive in increasing order with respect to their value), and it corresponds to a variable $X_i$ with $i \geq k+2$ if and only if all variables $j \geq k+2$ take the value $1/\varepsilon^{k+1}$. 

We define the random variable $Y$ as the number of variables among $X_{k+2}, \dots, X_{2(k+1)}$ that take the value $1/\varepsilon^{k+1}$. Then, conditioning on the value of $Y$, we have: 
\begin{eqnarray*}
\m{E}(X_{(k+1)}) &=& \sum_{j=0}^{k+1} \m{E}(X_{k+1} | Y=j) \m{P}(Y=j) \\
&=& 1 \cdot
(1- \varepsilon)^{k+1} + \sum_{j=1}^{k} \m{E}(X_{k+1} | Y=j) \cdot \m{P}(Y=j) +   \frac{1}{\varepsilon^{k+1}} \cdot \varepsilon^{k+1}\\
&=& (1- \varepsilon)^{k+1} + \sum_{j=1}^{k} \frac{1}{\varepsilon^{j} \binom{k+1}{j}} \cdot \binom{k+1}{j} \cdot \varepsilon^{j} \cdot (1- \varepsilon)^{k+1-j}+ 1  \\
&=&1+\sum_{j=0}^{k}  (1- \varepsilon)^{k+1-j}
\end{eqnarray*}

Let us now compute the optimal guarantee of the gambler. First, observe that the gambler should always accept the value $1/\varepsilon^{k+1}$, as it is the highest possible one. Furthermore, since the deterministic values are strictly increasing, if the gambler sees that a deterministic value has been removed, then all remaining variables—both deterministic and non-deterministic—must have either been removed or are equal to zero. In this case, the gambler receives 0. As a result, the gambler does not gain any useful information from observing past values, allowing us to restrict to strategies of the following form: (1) stop at time $i$, for some  $i \leq k+1$ (2) stop at the first positive value appearing after stage $k+2$. 

Call $ALG_i$ the payoff of a strategy of the form $(1)$. Under such a strategy, 
the gambler picks $X_i$ if and only if there are at least $i-1$ variables taking a value $1/\varepsilon^{k+1}$. That is, if $Y$ is greater than or equal to $i-1$. Then, under this strategy, the gambler obtains in expectation

\begin{eqnarray*}
    \m{E}(ALG_i) &=& \sum_{j=i-1}^{k+1} \m{E}(ALG_i | Y=j) \m{P}(Y=j) \\
    & = &  \frac{1}{\varepsilon^{i-1} \binom{k+1}{i-1}} \sum_{j=i-1}^{k+1} \binom{k+1}{j} \varepsilon^j (1-\varepsilon)^{k+1-j} \\ 
    &=& (1-\varepsilon)^{k-i+2} + \sum_{j=i}^{k+1} \frac{\binom{k+1}{j}}{\binom{k+1}{i-1}} \varepsilon^{j-i+1} (1-\varepsilon)^{k+1-j}.
\end{eqnarray*}

Last, consider strategy (2). This strategy gets a positive payoff if and only if all variables $X_{k+2}, \dots, X_{2(k+2)}$ are positive, which happens with probability $\varepsilon^{k+1}$. When this is the case, it gets payoff $\varepsilon^{-(k+1)}$. Consequently, (2) guarantees $\varepsilon^{k+1} \cdot \varepsilon^{-(k+1)}=1$. 

It follows that the optimal payoff of the gambler goes to 1 as $\varepsilon \rightarrow 0$. Moreover, we have $\m{E}(X_{(k+1)}) \rightarrow k+2$ as $\varepsilon \rightarrow 0$. Consequently, for each $\beta>0$, one can find $\varepsilon>0$ such that no algorithm achieves a competitive ratio larger than $1/(k+2)+\beta$ in the corresponding instance. This proves the theorem. 
\end{proof}

%\textcolor{red}{JOSE: Here we should probably include Bruno's example showing that $1/(2k+2)$ is tight for our algorithm}.

In what follows, we formally establish that the competitive ratio of $1/(2k+2)$ is tight for our proposed algorithm, $MSA_{RAND}$. \\

\begin{proposition} \label{prop:tight:NI}
For each $\beta>0$, there exists an instance with $k+2$ variables where $MSA_{RAND}$ does not achieve a better competitive ratio than $1/(2k+2)+\beta$. 
\end{proposition}
\begin{proof}
Let $X_1:=1$
and for $i\in \{ 2,\ldots,k+2  \} $
\[
X_i=\begin{cases}
    0 & \text{w.p. }1-\varepsilon,\\
    1/\varepsilon^{k+2} & \text{w.p. }\varepsilon,
\end{cases}
\]
with $\varepsilon \leq 1/2.$
We start by computing and estimating the $(k+1)$-max. 
\begin{eqnarray*}
\m{E}(X_{(k+1)}) &=& 1 \cdot (1-\varepsilon^{k+1})+\varepsilon^{-k-2} \cdot \varepsilon^{k+1},
%\varepsilon^{-1}.
\end{eqnarray*}
hence 
$\varepsilon^{-1} \leq \m{E}(X_{(k+1)}) \leq \varepsilon^{-1}+1$.

Let us now analyze $MSA_{RAND}$.
First, we show that for $i \geq 2$, $\m{E}(MSA_i)$ is $O(1)$ as $\varepsilon$ tends to 0. 

Let $i \in \{2, \dots, k+1\}$, and $A$ be the event ``the $i$-th largest sample is 0''. The probability of $A$ is at least $(1-\varepsilon)^{k+1}$,
since the latter corresponds to the probability that all samples from $F_2,\dots,F_{k+2}$ are 0. Then, $\m{P}(A)$ tends to 1 as $\varepsilon$ goes to 0. 

Assume that $A$ holds, meaning that the threshold for $MSA_i$ is 0. In this case, either $X_1$ is available and then $MSA_i$ picks it; or $X_1$ is not available and then the gambler is presented only with 0, getting a value 0.
%If the realization of $X_j$ is positive for all $j \geq 2$, then $X_1$ is available and $MSA_i$ picks it. Otherwise, all positive values are removed and the gambler is presented only with 0, getting a value 0. 
We deduce that
\begin{eqnarray*}
\m{E}(MSA_i|A)  \leq 1 \leq \varepsilon \m{E}(X_{(k+1)}).
\end{eqnarray*}
When $A$ is not realized, we use the rough upper bound $\m{E}(MSA_i|A^c) \leq \m{E}(X_{(k+1)})$. Therefore, we have 
\begin{eqnarray*}
\m{E}(MSA_i) &\leq&
\varepsilon \m{E}(X_{(k+1)}) \m{P}(A)+\m{E}(X_{(k+1)}) (1-\m{P}(A)).
\\
\end{eqnarray*}
%&\leq& \varepsilon \m{E}(X_{(k+1)})
%+[1-(1-\varepsilon)^{k+1}]\m{E}(X_{(k+1)})  \leq (k+1) \varepsilon 2 \varepsilon^{-1}=2(k+1). 
Since $\m{P}(A)$ converges to 1 as $\varepsilon$ tends to 0, we deduce that $\m{E}(MSA_i)/\m{E}(X_{(k+1)})$ converges to 0 as $\varepsilon$ goes to 0, as we wanted to show. 

It remains to evaluate the performance of $MSA_1$. To this end, let us define B the event ``the maximum sample is 1''. Note that this event occurs with probability $(1-\varepsilon)^{k+1}$.
 
 Assume that $B$ holds, meaning that the threshold for $MSA_1$ is 1. Then, 
either $X_1$ is available, and thus $MSA_1$ picks $X_1=1$ with probability $1/2$, and picks either $\varepsilon^{-k-2}$ or 0 with probability $1/2$; or, the gambler is presented only with 0.

Hence, 
\[\m{E}(MSA_1|B) \leq 1+ \varepsilon^{k+1} [\frac{1}{2}+\frac{1}{2} \varepsilon^{-k-2}] \leq 2+\frac{1}{2} \varepsilon^{-1} \leq (2\varepsilon+1/2) \m{E}(X_{(k+1)}).
\]

If B is not realized, we use again the inequality $\m{E}(MSA_1|B^c) \leq \m{E}(X_{(k+1)})$, obtaining that

\begin{eqnarray*}
\m{E}(MSA_1) &\leq&
(2\varepsilon+1/2) \m{E}(X_{(k+1)}) \m{P}(B)+\m{E}(X_{(k+1)}) (1-\m{P}(B)).
\\
\end{eqnarray*}

Since $\m{P}(B)$ converges to 1 as $\varepsilon$ tends to 0, we conclude that 
 %Consequently, either all the realizations of $X_i$, $i \geq 2$ are positive, and then $MSA_1$ picks $X_1=1$ with probability $1/2$, and $\varepsilon^{-k-2}$ otherwise; or, the gambler is presented only with 0.
%Hence, 
%$\m{E}(MSA_i) \leq \varepsilon^{k+1} [\frac{1}{2}+\frac{1}{2} \varepsilon^{-k-2}] \leq 1+\frac{1}{2} \varepsilon^{-1} \leq (\varepsilon+1/2) \m{E}(X_{(k+1)})$. We deduce that
%\begin{eqnarray*}
%\m{E}(MSA_{1}) &\leq& 2(k+1)+\frac{1}{k+1} \cdot \frac{1}{2} \varepsilon^{-1}
%\\
%&\leq& \left(\frac{1}{2k+2}+2(k+1) \varepsilon \right) \m{E}(X_{(k+1)}).
%\end{eqnarray*}
\[
\limsup_{\varepsilon \rightarrow 0} \m{E}(MSA_{RAND})/\m{E}(X_{(k+1)}) \leq \frac{1}{2k+2},\] and the result is proved. 
\end{proof}
\section{I.I.D.\ Case for $k=1$}\label{sect:iid}

In this section, we focus on i.i.d.\ instances where $F_1=\cdots = F_n$ in the case of $\RPI$ with $k=1$. That is, in the sequence $X_1,\ldots,X_n$, the maximum value has been removed. The main result of this section is the following: \\

\begin{theorem}\label{teo:iid:lb}
    For any information model, there is an algorithm for $1$-RPI with a competitive ratio of at least $0.4901$.
\end{theorem} 

For the rest of the section, we assume than $X_i$ are continuous with cdf $F(\cdot)$. Furthermore, following~\citep{perez2024optimal}, we can also assume that $F$ is strictly increasing and infinitely differentiable. Since the gambler observes the sequence of $n-1$ values, we can assume that the maximum of the $n$ values occurs in the last position $n$. Hence, the gambler faces the problem under the event $E=\{ X_1,\ldots,X_{n-1} < X_n \}=\{ \max_{i < n} X_i < X_n \}$. Note that $\PP(E)=1/n$ by the continuity of $F$.

To prove Theorem~\ref{teo:iid:lb}, we provide a fixed-threshold strategy that computes a threshold based on quantiles $q\in [0,1]$. That is, given $q\in [0,1] $, the algorithm computes $u\geq 0$ such that $q= \PP(X\geq \tau)=1-F(\tau)$, and accepts the first value at least $u$ in the observed sequence. We denote such an algorithm $ALG_q$. The following lemma provides a lower bound for a particular choice of quantiles $q$. \\

\begin{lemma}\label{lem:key_lemma_lower_bound}
    Let $n\geq 3$. For NI 1-RPI, if $ALG_q$ is run with $q=\alpha/(n-1)$ for $\alpha \in [0,2]$, then,
    \[
    \frac{\E (ALG_q)}{\EE (X_{(2)})} \geq \min\left\{ \frac{1-e^{-\alpha}}{\alpha} , 1-e^{-\alpha}(1+\alpha)  \right\},
    \]
    for any continuous cdf $F$.
\end{lemma}

Using this lemma, and by equating $(1-e^{-\alpha})/\alpha=1-e^{-\alpha}(1+\alpha)$, we obtain that $\alpha\approx 1.64718$ and the competitive ratio of fixed-threshold solutions is $\geq 0.4901$.

In Subsection~\ref{sub:upper_bound_single_threshold}, we show that no fixed-threshold solution can obtain a competitive ratio better than $0.5463$.

\subsection{Proof of Lemma~\ref{lem:key_lemma_lower_bound}}

In this subsection, we provide the lower bound on the competitive ratio of $ALG_q$ for $q=\alpha/(n-1)$. For notational convenience, we will avoid writing the subscript in $ALG$. The algorithm $ALG$ computes the threshold $u$ in advance and accept the first observed value that surpassed $u$. Then, the reward of $ALG$ as a function of $u$ is
\begin{align*}
   \E (ALG) = \sum_{i=0}^{n-2} \PP(X_1,\ldots,X_i < u\mid E) \EE\left[ X_{i+1} \mathbf{1}_{ \{X_{i+1}\geq u\}} \mid X_1,\ldots,X_i <u ,E  \right]
\end{align*}
By analyzing the different involved probabilities, we can find the following characterization of $\E(ALG)$ as a function of the quantile $q$: \\

\begin{proposition}\label{prop:characterization_iid}
    If $q\in [0,1]$, then
    \[
    \E(ALG)= \int_0^1 r(v) \sum_{i=0}^{n-2} \frac{n}{n-i-1}(1-q)^i \left(  \min\{ q, v \}  - \left(  \frac{ 1-  (1-\min\{q,v\})^{n-i}  }{n-i} \right)  \right) \, \mathrm{d}v,
    \]
    here $r(v)\geq 0$ is such that $F^{-1}(1-u)=\int_u^1 r(v) \, \mathrm{d}v$ which exists due to the assumptions over the cdf $F$.
\end{proposition}

The proof of this proposition is technical appears at the of the section. Likewise, we can find an expression for $\E(X_{(2)})$ in terms of $r(v)$ from the Proposition:
\[
\EE(X_{(2)})= \int_0^1 r(v) \PP(\mathrm{Binom}(n,v)\geq 2)\, \mathrm{d}v.
\]
Then,
\[
\frac{\E (ALG)}{\EE(X_{(2)})} \geq \inf_{v\in [0,1]} \left\{  \frac{\sum_{i=0}^{n-2} \frac{n}{n-i-1}(1-q)^i \left(  \min\{ q, v \}  - \left(  \frac{ 1-  (1-\min\{q,v\})^{n-i}  }{n-i} \right)  \right)   }{ \PP(\mathrm{Binom}(n,v)\geq 2  )}  \right\}
\]
This last bound is instance-independent and only depends on $n$ and $q$. Let $A_{n,q}(v)$ be the function in the infimum. We study the the infimum of $A_{n,q}(v)$ for the regime $v>q$ and $v<q$ separately. The following proposition characterizes the behavior of $A_{n,q}(v)$ in both regimes for $q=\alpha/(n-1)$ and $\alpha\leq 2$. We defer the proof to the end of the section. \\

\begin{proposition}\label{prop:key_property_of_A}
    For $q=\alpha/(n-1)$ and $\alpha\leq 2$, we have
    \begin{enumerate}
        \item If $v>q$, then, $A_{n,q}(v)$ is decreasing in $v$;

        \item If $v\leq q$, then, $A_{n,q}(v)$ is increasing in $v$.
    \end{enumerate}
\end{proposition}

With this proposition, we obtain
\begin{align*}
    \frac{\E(ALG)}{\EE(X_{(2)})} & \geq \min\left\{  \inf_{v\in [0,q]}\left\{ A_{n,q}(v) \right\}, \inf_{v\in [q,1]} \left\{  A_{n,q}(v) \right\}  \right\} \\
    & = \min \left\{ \lim_{v\to 0} A_{n,q}(v) , A_{n,q}(1)  \right\}\\
    & = \min\left\{   (n-1)\frac{1-(1-q)^{n-1}}{q} ,  1 - (1-q)^{n-1}(1+(n+1)q)   \right\} \\
    & \geq \min \left\{   \frac{1-e^{-\alpha}}{\alpha} ,  1 -e^{-\alpha} (1+\alpha)  \right\}
\end{align*}
where in the first equality we use Proposition~\ref{prop:key_property_of_A}, the second equality follows by a simple calculation, and in the last inequality we use the standard inequality $1-x\leq e^{-x}$. This finishes the proof of Lemma~\ref{lem:key_lemma_lower_bound} and by setting $q=\alpha/(n-1)$. This finishes the proof of the lemma.

We now provide the missing proofs  of Proposition~\ref{prop:characterization_iid} and~\ref{prop:key_property_of_A}.

\begin{proof}[Proof of Proposition~\ref{prop:characterization_iid}] The probability of reaching $i+1$ is the same as the probability of failing to observe a value at least larger than $u$ among the first $i$ observations, which is given by $\PP(X_1,\ldots,X_i < u\mid E)$, while the reward at $i+1$ is the expected value when $X_{i+1}\geq u$. For the sake of notation, we define $E_{i,u}=\{  X_1,\ldots,X_i <u \}$ for $i=0,\ldots,n-2$. We need to compute $\PP(E_{i,u} \mid E)$ and $\PP(X_{i+1}<u'\mid  E_{i,u}, E )$ for $u'\geq u$.

Clearly $\PP(X_1<u,\ldots,X_i < u \mid E)=1$ for $i=0$ so let's assume that $i>0$. Then,
\begin{align*}
    \PP(E_{i,u}\mid E) & = n \int_0^\infty \PP(X_1,\ldots,X_i <u , X_1,\ldots,X_{n-1}<x) \, \mathrm{d}F(x)\\
    & = n \int_0^u F(x)^{n-1}\, \mathrm{d}F(x) + n \int_{u}^{\infty} F(u)^i F(x)^{n-i-1}\, \mathrm{d}F(x)  \\
    & = F(u)^n + n F(u)^i \left( \frac{1- F(u)^{n-i}}{n-i} \right) 
\end{align*}
Note that if we assume that $0^0=1$, then the formula above also work in the case $i=0$.

Now, for $u'\geq u$,
\begin{align*}
    \PP(X_{i+1}<u'\mid E_{i,u},E) & = n\frac{\PP(X_{i+1}< u',E_{i,u},E)}{\PP(E_{i,u}\mid E)}
\end{align*}
we already computed the denominator; hence, we focus on computing the numerator.
\begin{align*}
    n\PP(X_{i+1}<u' , E_{i,u},E ) & = n \int_0^{\infty} \PP(X_{i+1}<u', X_1,\ldots,X_i<u,X_1,\ldots,X_{n-1}<x) \, \mathrm{d}F(x)\\
    & = n\int_0^u F(x)^{n-1}\, \mathrm{d}F(x) + n\int_{u}^{u'} F(u)^{i} F(x)^{n-i-1}\, \mathrm{d}F(x)\\
    & \quad + n F(u)^i F(u') \int_{u'}^\infty F(x)^{n-i-2}\,\mathrm{d}F(x)\\
    & = F(u)^n + n F(u)^i \left( \frac{F(u')^{n-i} - F(u)^{n-i}}{n-i}  \right) \\
    &\quad + nF(u)^i F(u)' \left( \frac{1- F(u')^{n-i-1}}{n-i-1} \right) \\
    & = F(u)^n - \frac{n}{n-i}F(u)^n + \frac{n}{n-i}F(u)^i F(u')^{n-i}\\
    & \quad +\frac{n}{n-i-1}F(u)^i F(u') - \frac{n}{n-i-1}F(u)^i F(u')^{n-i}
\end{align*}
Then,
\begin{align*}
    \frac{\mathrm{d}}{\mathrm{d}x} \PP(X_{i+1}<x\mid E_{i,u},E) & = F(u)^i \left(\frac{n}{n-i-1}\right)\frac{ 1- F(x)^{n-i-1}}{\PP(E_{i,u}\mid E)}  \frac{\mathrm{d}F}{\mathrm{d}x}(x)
\end{align*}
and
\begin{align*}
    \mathbf{E}\left[ X_{i+1}\mathbf{1}_{\{X_{i+1}\geq u  \} }\mid E_{i,u}, E \right] & = \frac{1}{\PP(E_{i,u}\mid E)} F(u)^i \frac{n}{n-i-1}\int_{u}^\infty x\left( 1-F(x)^{n-i-1}  \right)\, \mathrm{d}F(x).
\end{align*}
Then,
\begin{align*}
\E (ALG) &= \sum_{i=0}^{n-2} \frac{n}{n-i-1} F(u)^i \int_u^\infty x \left( 1-F(x)^{n-i-1} \right) \, \mathrm{d}F(x)   \\
& = \sum_{i=0}^{n-2} \frac{n}{n-i-1} (1-q)^i \int_0^q F^{-1}(1-w)(1- (1-w)^{n-i-1})\, \mathrm{d}w \tag{Change of variable $1-q=F(u)$} \\
& = \sum_{i=0}^{n-2} \frac{n}{n-i-1}(1-q)^i \int_0^q \int_{w}^1 r(v) \, \mathrm{d}v (1 - (1-w)^{n-i-1} )\, \mathrm{d}w\tag{Using that $F^{-1}(1-w)$ is strictly decreasing and differentiable}\\
& = \int_0^1 r(v) \sum_{i=0}^{n-2} \frac{n}{n-i-1}(1-q)^i \left(  \min\{ q, v \}  - \left(  \frac{ 1-  (1-\min\{q,v\})^{n-i}  }{n-i} \right)  \right) \, \mathrm{d}v
\end{align*}
\end{proof}

\medskip

\begin{proof}[Proof of Proposition~\ref{prop:key_property_of_A}]

For $v>q$, we have
\[
A_{n,q}(v) = \frac{\sum_{i=0}^{n-2} \frac{n}{n-i-1}(1-q)^i \left( 
 q - \left( \frac{1-(1-q)^{n-i}}{n-i}  \right) \right)   }{\PP(\mathrm{Binom}(n,v)\geq 2)} = \frac{1-(1-q)^{n-1}(1+(n-1)q)}{\PP(\mathrm{Binom}(n,v)\geq 2)}
\]
This last function is decreasing in $v$ attaining its minimum at $v=1$. 

For $v\leq q$, we have
\begin{align*}
A_{n,q}(v) &= \frac{\sum_{i=0}^{n-2} \frac{n}{n-i-1}(1-q)^i \left( 
 v - \left( \frac{1-(1-v)^{n-i}}{n-i}  \right) \right)   }{\PP(\mathrm{Binom}(n,v)\geq 2)} \\
 & = \sum_{i=0}^{n-2} \frac{n}{(n-i-1)(n-i)}(1-q)^i \left(\frac{
 (1-v)^{n-i} - (1 - (n-i)v)}{\PP(\mathrm{Binom}(n,v)\geq 2)} \right)\\
 & = \sum_{i=2}^{n} \frac{n}{i(i-1)}(1-q)^i G_{n,n-i}(v),
\end{align*}
where $G_{n,i}(v) = \frac{(1-v)^i -  (1-iv)}{\PP(\mathrm{Binom}(n,v)\geq 2)}$ for $i\in \{ 2,\ldots,n\}$. To conclude that $A_{n,q}(v)$ is increasing in $v$, it is enough to show that $G_{n,i}(v)$ is increasing in $v$ for all $i\in \{2,\ldots,n\}$. Then,
\[
G_{n,i}'(v) = \frac{ g_{n,i}(v) }{\PP(\mathrm{Binom}(n,v)\geq 2)^2},
\]
where $g_{n,i}(v)= (-i(1-v)^{i-1} +i)\PP(\mathrm{Binom}(n,v)\geq 2) -  ((1-v)^i - (1-iv))n(n-1)v(1-v)^{n-2} $. To conclude the proof, it is enough to show that $g_{n,i}(v)\geq 0$. We note that $g_{n,i}(0)=0$, so we only need to prove that $g_{n,i}'(v)\geq 0$. Now, 
\begin{align*}
    g_{n,i}'(v) & = i(i-1)(1-v)^{i-2}\PP(\mathrm{Binom}(n,v)\geq 2) - ((1-v)^i - (1-iv))n(n-1)(1-v)^{n-2}\\
    & \quad + ( (1-v)^i - (1-iv)  ) n (n-1)(n-2)v(1-v)^{n-3} \\
    & = i(i-1)(1-v)^{i-2}\PP(\mathrm{Binom}(n,v)\geq 2) \\
    &\quad -n(n-1)(1-(n-1)v)(1-v)^{n-3}( (1-v)^i - (1-iv) )
\end{align*}
Using the second equality it is easy to verify that $g_{n,i}'(v)\geq 0$ for $v\geq 1/(n-1)$. Hence, from now, we assume that $v< 1/(n-1)$. Furthermore, by inspection, we can verify that $g_{n,i}'(v)\geq 0$ for $n\in \{3,4\}$ and $i\in \{2,\ldots,n \}$; hence, from now on, we assume that $n\geq 5$. The following claim allows us to focus only on lower bounding $g_{n,2}'(v)$.

\begin{claim}\label{claim:monotonicity_g}
    For $n\geq 5$, $v\leq 1/(n-1)$ and for all $i\in \{ 2,\ldots,n-1\}$, we have $g_{n,i}'(v)\leq g_{n,i+1}'(v)$.
\end{claim}

This proof requires lower bounding several polynomials and various case analysis; hence, we defer it to the end. Now, note that
\begin{align*}
    g_{n,2}'(v) & = 2\PP(\mathrm{Binom}(n,v)\geq 2) - n(n-1)(1-(n-1)v)(1-v)^{n-3}v^2 \\
    & \geq n(n-1)v^2 (1-v)^{n-2} - n(n-1)(1-(n-1)v)v^2 (1-v)^{n-3}\\
    & = n(n-1)v^2(1-v)^{n-3}( 1-(1-v)(1-(n-1)v) )\\
    & = n(n-1)v^2 (1-v)^{n-3}(v+(n-1)v(1-v))
\end{align*}
where in the first inequality we use the lower bound $\PP(\mathrm{Binom}(n,v)\geq 2)\geq \PP(\mathrm{Binom}(n,v)=2)$. From here, we obtain that $g_{n,2}'(v)\geq 0$ with equality at $v=0$. Using Claim~\ref{claim:monotonicity_g}, we conclude that $G_{n,i}'(v)\geq 0$ for all $i\in \{2,\ldots,n\}$.

\begin{proof}[Proof of Claim~\ref{claim:monotonicity_g}]
    Indeed,
\begin{align}
    g_{n,i}'(v)- g_{n,i+1}'(v) &= -\PP(\mathrm{Binom}(n,v)\geq  2) \left(  2 -(i+1)v   \right)i(1-v)^{i-2} \\
    & \quad + v(1-(1-v)^i) (1-(n-1)v) n (n-1) (1-v)^{n-3}\nonumber \\
    & \leq (1-v)^{i-3}\Big( -i\left(  2 -(i+1)v   \right)(1-v)\PP(\mathrm{Binom}(n,v)\geq  2) \big.\nonumber \\
    &\quad \big. + v(1-(1-v)^i)(1-(n-1)v)n(n-1) \Big) \label{ineq:difference_g}
\end{align}
where in the inequality we use that $(1-v)^{n-3}\leq (1-v)^{i-3}$. Now, let $\phi_{n,i}(v)$ be the term in the big parenthesis in~\eqref{ineq:difference_g}, so $g_{n,i}'(v)-g_{n,i+1}'(v)\leq (1-v)^{i-3}\phi_{n,i}(v)$. We now focus on proving that $\phi_{n,i}(v)\leq 0$. For this, we will reduce the problem to bounding only $\phi_{n,2}(v)\leq 0$ by proving that $\phi_{n,2}(v) \geq \phi_{n,i}(v)$ for all $i=2,\ldots,n-1$. To prove this last inequality, we analyze the difference $\phi_{n,i}(v)-\phi_{n,i+1}(v)$ for $i\in \{2,\ldots,n-2 \}$:
\begin{align*}
    &\quad \phi_{n,i}(v)-\phi_{n,i+1}(v) \\
    & = 2(1-v)(1-(i+1)v)\PP(\mathrm{Binom}(n,v)\geq 2)-v^2(1-v)^i(1-(n-1)v)n(n-1) \\
    & \geq v^2 (1-v)^i(n-1)n\left( (1-(i+1)v)(1-v)^{n-i-2}\left( 1+\frac{n-5}{3} v \right) -(1-v(n-1)) \right) \\
    & = v^2 (1-v)^i (n-1)n \cdot\theta_{n,i}(v).
\end{align*}
Note that the function $(1-(i+1)v)(1-v)^{n-i-2}$ is decreasing in $i$; hence, for $n\geq 5$, we have $\theta_{n,i}(v)\geq \theta_{n,n-2}(v)$. From here, we obtain
\[
\phi_{n,i}(v)- \phi_{n,i+1}(v) \geq v^2(1-v)^i (n-1)n \cdot\theta_{n,n-2}(v).
\]
On the other hand, we have $\theta_{n,n-2}(v)=(1-(n-1)v)\left(  
\frac{n-5}{3}\right)v\geq 0$. From here, we obtain that $\phi_{n,i}(v)\geq \phi_{n,i+1}(v)$ for all $i\in \{2,\ldots,n-2 \}$ and so $\phi_{n,2}(v)\geq \phi_{n,i}(v)$ for all $i\in \{2,\ldots,n-2\}$. Now,
\begin{align*}
    \phi_{n,2}(v) & = -2 (2-3v)(1-v) \PP(\mathrm{Binom}(n,v)\geq 2) + v^2(2-v)(1-(n-1)v)n(n-1)  \\
    & \leq -2(2-3v)(1-v)\left(  \binom{n}{2}v^2(1-v)^{n-2} + \binom{n}{3}v^3(1-v)^{n-3}    \right) \\
    & \quad + v^2 (2-v)(1-(n-1)v)n(n-1)\\
    & = n(n-1)v^2 \left( -(2-3v) \left((1-v)^{n-1}+\frac{n-2}{3}v(1-v)^{n-2}\right)+(2-v)(1-(n-1)v)   \right)\\
    & = n(n-1)v^2\left( (2-v)(1-(n-1)v) - (2-3v) (1-v)^{n-2} \left(1+\frac{n-5}{3}v\right)  \right)
\end{align*}
We analyze this last bound for the case $n=5$ and case $n\geq 6$ separately. For $n=5$, we have
\[
\phi_{5,2}(v) \leq 20v^2\left( (2-v)(1-4v) -(2-3v)(1-v)^{3}   \right).
\]
The polynomial $(2 - v)(1 - 4v) - (2 - 3v)(1 - v)^3$ has roots $v\in \{ 0, (11-i\sqrt{11})/6, (11+i\sqrt{11})/6 \}$ with $0$ having multiplicity $2$. Since, the polynomial tends to $-\infty$ when $v\to\infty$ and its only $0$ when $v=0$, we deduce that $\phi_{5,2}(v)\leq 0$.

Now, assume that $n\geq 6$, then
\begin{align*}
    \phi_{n,2}(v) & \leq n(n-1)v^2 \left( (2-v)(1-v)^{n-1} - (2-3v) (1-v)^{n-2} \left(1+\frac{n-5}{3}v\right)  \right)\\
    & = n(n-1) v^2 \left( (1-v)^{n-2} \cdot v \cdot\left( -\frac{2(n-5)}{3} + (n-4)v \right)     \right)
\end{align*}
where in the first inequality we use Bernoulli's inequality on $1-(n-1)v \leq (1-v)^{n-1}$ and in the equality we simply reorder the big parenthesis from the previous line. Now, the polynomial $v(-2(n-5)/3 +(n-4)v)$ has $2$ roots at $v\in \{0, 2(n-5)/(3(n-4))$. Hence, for $v\leq 1/(n-1)$, we must have that $\phi_{n,2}(v)\leq 0$.
 
Going back to the function $g_{n,i}'$, all our calculations give us
\begin{align*}
    g_{n,i}'(v)- g_{n,i+1}'(v) &\leq (1-v)^{i-3} \phi_{n,2}(v) \leq 0,
\end{align*}
which finishes the proof of Claim~\ref{claim:monotonicity_g}.
%end of claim
\end{proof}
%end of proposition
\end{proof}

\subsection{An Upper Bound for Single-Threshold 
Solutions}\label{sub:upper_bound_single_threshold}

We present an instance that shows that no strategy in the class of single-threshold (including randomization) can obtain a competitive ratio larger than $0.5463$. We use a counterexample motivated by~\cite{perez2025iid}. For $n\geq 1$, we consider the following function from $(0,1]$ to $\mathbf{R}_+$
\[
f(u) = \frac{a \cdot c_n}{u} \mathbf{1}_{(0,1/n^{10})}(u) + b\cdot \mathbf{1}_{[1/n^{10},\beta/n]}(u)
\]
where $\mathbf{1}_X$ is the indicator function that is $1$ for $u\in X$ and $0$ for $u\notin X$, $a,b>0$ and $\beta>1/n$ are positive constants to be optimized and $c_n= \left( n \cdot \left(1- \left( 1 - 1/n^{10} \right)^{n-1} \right) \right)^{-1}$. We are going to assume that $a+b\leq 1$. For $n$ large enough, we have that $f$ is nonincreasing. 

Now, we can construct a random variable from $f$ as follows. First, we add a small perturbation to $f$ so $f$ is smooth and strictly decreasing. This can be done by taking a convolution with a smooth function. Let's call $f_\varepsilon$ the resulting function, with small error $\varepsilon>0$; hence, when $\varepsilon\to 0$, we have $f_{\varepsilon}(u)\to f(u)$, except for a set of measure $0$. Note that $f_{\varepsilon}$ is surjective in $\R_+$. Now, for $x\geq 0$, let $F_\varepsilon(x)= 1-f_\varepsilon^{-}(x)$. Note that $F$ is increasing, $F_\varepsilon(0)=0$ and $F_\varepsilon(+\infty)=1$; hence, $F_\varepsilon$ is a valid CDF. We define the random variable $X_\varepsilon$ to be the random variable following $F_\varepsilon$. By construction, $F_\varepsilon^{-1}(1-u)=f_\varepsilon(u)$. For $\varepsilon\to 0$, we have $F_\varepsilon^{-1}(1-u)\to f(u)$, except for a set of measure $0$ in $[0,1]$. To avoid notational clutter, from now on, we simply work with $f(u)$ instead of $f_\varepsilon$. By abusing notation, we will write $F^{-1}(1-u)=f(u)$, but it has to be understood that this equality occurs except for the points $1/n^{10}$ and $\beta/n$.

We now consider a sequence of $n$ independent random variables following $F$ (the limit of $F_\varepsilon$ when $\varepsilon\to 0$). We assume that $n$ is large. The result now follows from the following two Lemmas. \\

\begin{lemma}\label{lem:limit_1_threshold_opt}
    We have $\EE (X_{(2)}) \to a+b(1-e^{-\beta}(1+\beta))$ when $n\to \infty$. \\
\end{lemma}

\begin{lemma}\label{lem:general_bound_1_threshold_algs}
    There is $n_0\geq 0$ such that for any algorithm $\ALG$, if the input is of length $n\geq n_0$, the value collected by the algorithm is bounded as
    \[
    \E(\ALG) \leq p(a,b,\beta) + 5\frac{\beta(1+\beta)^2}{n-\beta},
    \]
    where $p(a,b,\beta)= \max_{\lambda \in [0,\beta] }\left\{ a ({1-e^{-\lambda}})/{\lambda} + b \left( 1 - e^{-\lambda}(1+\lambda) \right)  \right\}$.
\end{lemma}

We first provide the tight upper bound and then we prove the lemmas. Using these two lemma, for any algorithm and $n\geq n_0$, we have
\begin{align*}
    \frac{E(\ALG)}{\EE (X_{(2)})} & \leq \frac{p(a,b,\beta) + 5{\beta(1+\beta)^2}/{(n-\beta)}}{\EE (X_{(2)}) }
\end{align*}
Using numerical optimization to minimize $p(a,b,\beta)$, we found $a\approx 0.5463 $, $b\approx 0.4537$ and $\beta\approx 109.131$, we obtain $p(a,b,\beta)\approx 0.5463$. Hence, for $n$ large, we obtain that $\E (\ALG)/\EE (X_{(2)}) \leq 0.5463+ o(n)$. This shows that with one threshold, we cannot recover the approximation of $1-1/e\approx 0.6321$ in the standard prophet inequality.
%e^{-\alpha^*})/\alpha^*$ and $p(\alpha^*)= 1$
%a -> 0.546376, 
%b -> 0.453624, 
%beta -> 109.131

In the remainder of the subsection, we present the proof of Lemma~\ref{lem:limit_1_threshold_opt} and~\ref{lem:general_bound_1_threshold_algs}.

\begin{proof}[Proof of Lemma~\ref{lem:limit_1_threshold_opt}]
    We have
    \begin{align*}
        \EE (X_{(2)}) & = n(n-1)\int_0^1 F^{-1}(1-u)q (1-q)^{n-2}\, \mathrm{d}q\\
        & = n(n-1)\int_0^{1/n^{10}} a\cdot c_n (1-q)^{n-2}\, \mathrm{d}q + n(n-1)\int_{1/n^{10}}^{\beta/n}  b q(1-q)^{n-2}\, \mathrm{d}q \\
        & = a c_n  n \left( 1 - \left( 1 - \frac{1}{n^{10}} \right)  \right)\\
        & \quad + b \left( \frac{1}{n^9}\left( 1 - \frac{1}{n^{10}} \right)^{n-1} - \beta\left( 1- \frac{\beta}{n} \right)^{n-1} + \left( 1- \frac{1}{n^{10}} \right)^n - \left( 1-\frac{\beta}{n} \right)^{n} \right)
    \end{align*}
    The conclusion now follows by taking limit in the last equality.
\end{proof}

\begin{proof}[Proof of Lemma~\ref{lem:general_bound_1_threshold_algs}]
    We can parametrize every single-threshold algorithm via the quantile chosen by it. If $\ALG_q$ denotes the value obtained by a single-threshold algorithm that always uses quantile $q$, we have $\ALG \leq \max_{q\in [0,1]}\ALG_q$. We analyze this last maximum for $q\leq 1/n^{10}$, $q\in [1/n^{10},\beta/n]$ and $q\geq \beta/n$.

    For $q\leq 1/n^{10}$, we have
    \begin{align*}
        \E(\ALG_q) & = \sum_{k=1}^{n-1}\frac{n}{k}(1-q)^{n-k-1}\int_0^{q}\frac{ a c_n }{w} \left( 1-(1-w)^k \right)\, \mathrm{d}w \\
        & \leq a \cdot c_n \sum_{k=1}^{n-1} \frac{n}{k} \sum_{\ell=0}^{k-1} \int_0^{1/n^{10}} (1-w)^\ell \,\mathrm{d}w\\
        & = a c_n \sum_{k=1}^{n-1}\sum_{\ell=0}^{k-1} \frac{1-(1-1/n^{10})^{\ell+1}}{\ell+1}\\
        & \leq a c_n \sum_{k=1}^{n-1}\frac{n}{k} \frac{k}{n^{10}}\\
        & = a \cdot \frac{1/n^{9}}{1- (1-1/n^{10})^{n-1}} \leq a \left( 1 +  \frac{3}{n-1}\right)
    \end{align*}
    where in the first inequality we upper bounded the integral for $q=1/n^{10}$ and we also upper bounded $(1-q)^{n-k-1}\leq 1$; in the second equality we performed the integration; in the second inequality we used Bernoulli's inequality: $(1-1/n^{10})^{\ell+1}\geq 1-(\ell+1)/n^{10}$, and in the last inequality we used the following claim.
    \begin{claim}
        We have $c_n/n^8\leq \left( 1 + 3/(n-1)  \right)$.
    \end{claim}

    \begin{proof}
        First, note that 
        \begin{align}
        \left( 1 - \frac{1}{n^{10}} \right)^{n-1}\leq e^{-(n-1)/n^{10}} \leq 1 - \frac{n-1}{n^{10}} + \frac{(n-1)^2}{2n^{20}} \label{ineq:bound_for_cn}    
        \end{align}
        Then,
        \begin{align*}
            \frac{c_n}{n^8} & = \frac{1/n^9}{1-(1-1/n^{10})^{n-1}} \\
            & \leq \frac{n}{(n-1)\left( 1 - (n-1)/(2n^{10}) \right)}\\
            & \leq \left(1 + \frac{1}{n-1} \right)\left(1 + \frac{1}{2n} \right)\\
            & \leq 1 + \frac{3}{n-1}
        \end{align*}
        where in the first inequality we used inequality~\eqref{ineq:bound_for_cn}, the second inequality follows by $1\leq (1-(n-1)/2n^{10})(1+1/2n)$ and the last inequality follows simply by expanding the multiplication and bounding $ 1/2n, 1/(2n(n-1))\leq 1/(n-1)$.
    \end{proof}

    Therefore, for any $q\leq 1/n^{10}$, we have that the value of the algorithm is upper bounded by $a\cdot (1 + 3/(n-1))$.

    For $q=\lambda/n\in [1/n^{10},\beta]$, we have
    \begin{align*}
        \E(\ALG_q) & = \sum_{k=1}^{n-1} \frac{n}{k}(1-q)^{n-k-1}\int_0^{1/n^{10}} \frac{a c_n}{w}\left( 1-(1-w)^k \right) \, \mathrm{d}w\\
        & \quad + \sum_{k=1}^{n-1}\frac{n}{k}(1-q)^{n-k-1} \int_{1/n^{10}}^{q}b \left( 1 - (1-w)^k \right)\, \mathrm{d}w\\
        & \leq a c_n \sum_{k=1}^{n-1}\frac{n}{k} (1-q)^{n-k-1} \frac{k}{n^{10}} \\
        & \quad + b \sum_{k=1}^{n-1} \frac{n}{k}(1-q)^{n-k-1} \left(\frac{(1-q)^{k+1}- ( (1-1/n^{10})^{k+1} -  (q-1/n^{10})(k+1) ) }{k+1}\right)    \end{align*}
    where in the first inequality we upper bounded the first term in the first line in a manner similar to the case $q\leq 1/n^{10}$ and we integrated the second term. The following two claims allow us to upper bound this last inequality by controlling the error.   
    \begin{claim}\label{claim:limit_of_cn}
        We have $c_n\left(  1 - (1-q)^{n-1} \right)/(qn^9) \leq (1-e^{-\lambda})/\lambda + \beta(1+\beta)/(n-\beta)$
    \end{claim}

    \begin{proof}
    In the proof, we use that $\lambda \leq \beta \leq 2$. First, we note that
    \begin{align*}
        1-\left( 1 - \frac{\lambda}{n} \right)^{n-1}& = 1-e^{-\lambda} + e^{-\lambda} - \left(1 - \frac{\lambda}{n} \right)^{n-1}\\
        & \leq 1 - e^{-\lambda} +e^{-\lambda} - e^{-\lambda\left( \frac{n-1}{n-\lambda} \right)}\\
        & = 1-e^{-\lambda}+ e^{-\lambda}\left( 1 - e^{\lambda\left( \frac{1-\lambda}{n-\lambda}\right)}  \right) \\
        & \leq 1-e^{-\lambda} + \frac{|\lambda(1-\lambda)|}{n-\lambda}\leq 1- e^{-\lambda} + \frac{\beta(1+\beta)}{n-\beta}
    \end{align*}
    where in the first inequality we used that $1/(1-\lambda/n)\leq e^{\lambda/(n-\lambda)}$ using the standard inequality $1+x\leq e^x$, in the second inequality we used that $e^{-\lambda}\leq 1$ and $1-|x|\leq e^{-x}$, and in the last inequality, we simply used that $\lambda \leq \beta$.
    \end{proof}

    \begin{claim}
        We have
        \begin{align*}
            &\sum_{k=1}^{n-1} \frac{n}{k}(1-q)^{n-k-1} \left(\frac{(1-q)^{k+1}- ( (1-1/n^{10})^{k+1} -  (q-1/n^{10})(k+1) ) }{k+1}\right) \\
            & \leq 1-e^{-\lambda}(1+\lambda) + 3\frac{\beta(1+\beta)^2}{n}.
        \end{align*}
    \end{claim}
    \begin{proof}
        First, we have
        \begin{align*}
            \sum_{k=1}^{n-1} \frac{n}{k(k+1)}(1-q)^{n-k-1}\left( (1-q)^{k+1} - (1-q(k+1)) \right) & = \sum_{k=1}^{n-1} n q^2 (1-q)^{n-k-1} \\
            & \quad -(1-q)^n + (1-q n) \\
            & = 1 - (1-q)^{n-1}(1-q+qn)
        \end{align*}
        where in the first equality we simply rearranged the sum and in the next line we computed the sum. Then,
        \begin{align*}
            &\, \sum_{k=1}^{n-1} \frac{n}{k}(1-q)^{n-k-1} \left(\frac{(1-q)^{k+1}- ( (1-1/n^{10})^{k+1} -  (q-1/n^{10})(k+1) ) }{k+1}\right)\\
            & \leq \sum_{k=1}^{n-1} \frac{n}{k(k+1)}(1-q)^{n-k-1}\left( (1-q)^{k+1} - (1-q(k+1)) \right) \\
            & \quad +\sum_{k=1}^{n-1} \frac{n}{k(k+1)}(1-q)^{n-k-1}\left(  1 - (1-1/n^{10})^{k+1} \right) \\
            & \leq 1 - (1-q)^{n-1}(1-q+qn) + \sum_{k=1}^{n-1} \frac{1}{k} \frac{1}{n^{9}}\\
            & \leq 1 - (1-q)^{n-1}(1-q+qn) + \frac{\ln (n)}{n^9}
        \end{align*}
        where in the second inequality, we added and subtracted $1$ to the parenthesis to form the term studied at the beginning of the proof, in the second inequality, we used Bernoulli's inequality and in the last inequality, we used the standard Harmonic sum bound. Finally, 
        \begin{align*}
            -(1-q)^{n-1}(1-q+qn) & = -(1+\lambda)\left(1 - \frac{\lambda}{n} \right)^{n-1} + \frac{\lambda}{n} \left(1 - \frac{\lambda}{n} \right)^{n-1}\\
            & \leq -(1+\lambda)e^{-\lambda\left( \frac{n-1}{n-\lambda}\right)} + \frac{\beta}{n} \\
            & = \left( e^{-\lambda} - e^{-\lambda\left( \frac{n-1}{n-\lambda} \right)} \right)(1+\lambda) - e^{-\lambda}(1+\lambda) + \frac{\beta}{n} \\
            & \leq -e^{-\lambda}(1+\lambda)+ \frac{\beta(1+\beta)^2}{n-\beta} + \frac{\beta}{n} 
        \end{align*}
        where in the second and last inequalities, we used the same bounds used in the analysis of Claim~\ref{claim:limit_of_cn}. From here, the inequality of the claim follows.
    \end{proof}
    Therefore,
    \begin{align*}
        \E(\ALG_q) & \leq a\cdot \left( \frac{1-e^{-\lambda}}{\lambda} \right) + b (1-e^{-\lambda}(1+\lambda)) + 4 \frac{\beta(1+\beta)^2}{n-\beta}
    \end{align*}
    For $q\geq \beta/n$, we have
    \begin{align*}
        \E(\ALG_q) & = \sum_{k=1}^{n-1} \frac{n}{k}(1-q)^{n-k-1}\int_0^{1/n^{10}} \frac{a c_n}{w}\left( 1-(1-w)^k \right) \, \mathrm{d}w\\
        & \quad + \sum_{k=1}^{n-1}\frac{n}{k}(1-q)^{n-k-1} \int_{1/n^{10}}^{\beta/n}b \left( 1 - (1-w)^k \right)\, \mathrm{d}w.
    \end{align*}
    This last term is decreasing in $q$; hence it attains its maximum at $q=\beta/n$.

    The conclusion of the lemma follows by putting together the three bounds that we found. Additionally, the bound for $\lambda\in [1/n^{9},\beta]$ supersedes the bound for $q\leq 1/n^{10}$ and $q\geq \beta/n$. 
\end{proof}

\section{Conclusion and Final Remarks}

In this work, we introduced the residual prophet inequality problem (\RPI), a new variant of the classical prophet inequality model where the top $k$ variables are removed before observation. Our formulation highlights the impact of correlation in sequential selection problems and demonstrates that classical single-threshold approaches are insufficient in this setting. We provided a randomized algorithm with a competitive ratio of $1/(k+2)$ for the FI model and showed the tightness of this bound. For the NI model, we give a randomized threshold algorithm with a competitive ratio of $1/(2k+2)$. Additionally, we analyzed the i.i.d.\ case of 1-RPI and proposed an algorithm with a competitive ratio of at least 0.4901. Furthermore, we proved that no single-threshold strategy can achieve a competitive ratio greater than 0.5464. 

Since this is the first time $\RPI$ is introduced, our work opens up several promising directions for future research. One such direction is to investigate whether the $1/(k+2)$ competitive ratio can be achieved using a threshold-based strategy. Another natural avenue is to determine the tight competitive ratio for single-threshold strategies in the i.i.d. case of 1-RPI, and to explore how these results might extend to $\RPI$ for $k \geq 2$. One of the limitations of our current analysis is that it relies heavily on being able to compute probabilities under the condition of the maximum value being removed; these probabilities become intractable to handle for larger values of $k$. Naturally, determining the optimal policy $k$-RPI or even analyzing multi-thresholds strategies are exciting future questions, even for $k=1$.

A natural extension is to explore $\RPI$ under natural combinatorial constraints such as cardinality or matroid constraints, where the gambler can select multiple values while satisfying feasibility conditions, as it has been studied for the classical prophet inequality by~\cite{kleinberg2012matroid}. 

The $\RPI$ problem is very pessimistic as the $k$ largest random variables are removed from the observed sequence. A more relaxed model would consider probabilities of failure. For example, a possibility could be where the $i$-th largest variable is removed with probability $p_i$. This is related to the model by~\cite{perez2024robust,smith1975secretary} and \cite{tamaki1991secretary} where $p_i=p$ for all $i$. 

 Finally, an interesting extension is to study if better competitive ratios for $\RPI$ can be obtained when the removed variables are not necessary the largest, and the gambler has some offline information regarding the variables and/or the values removed.

\backmatter

% \bmhead{Supplementary information}

% If your article has accompanying supplementary file/s please state so here. 

% Authors reporting data from electrophoretic gels and blots should supply the full unprocessed scans for key as part of their Supplementary information. This may be requested by the editorial team/s if it is missing.

% Please refer to Journal-level guidance for any specific requirements.
%\section*{Statements and Declarations}

%\noindent {\bf{Acknowledgments}}\\
 
%\noindent {\bf{Ethical Approval}}
%Ethical approval not applicable to this article. \\

%\noindent {\bf{Competing interests}} 
%The authors have no competing interests to declare.\\
 
%\noindent {\bf{Authors' contributions}} 
%All authors contributed equally to this work. \\
 
\subsection*{Funding}
This work was partially supported by ANID Chile through grants FB210005 (CMM), 11240705 (FONDECYT), and AFB230002 (ISCI); 
by ANR France grants {ANR-21-CE40-0020} (CONVERGENCE), and {ANR-17-EURE-0010} (Investissements d'Avenir program).
%\\
 
%\noindent {\bf{Availability of data and materials}} This declaration is not applicable to this article.

%\bmhead{Acknowledgments}
%This work was partially supported by ANID Chile through grants FB210005 (CMM), 11240705 (FONDECYT), and AFB230002 (ISCI); 
%by ANR France grants {ANR-21-CE40-0020} (CONVERGENCE), and {ANR-17-EURE-0010}.

\bibliography{sn-bibliography.bib}

\end{document}